\documentclass{techrep}

\usepackage{lineno}
\usepackage{amsfonts}
\usepackage{amsthm}

\newcommand{\comment}[1]{}

\newtheorem{theorem}{Theorem}{\bfseries}{\itshape}

{\bfseries}{\itshape}

{\bfseries}{\itshape}

{\bfseries}{\itshape}

{\bfseries}{\itshape}

\newtheorem{obs}{Observation}{\bfseries}{\itshape}

\newtheorem{clm}{Claim}{\bfseries}{\itshape}

\newtheorem{assumption}{Assumption}{\bfseries}{\rm}

\theoremstyle{definition}

\newtheorem{definition}{Definition}

{\bfseries}{\rm}

\theoremstyle{plain}

\newenvironment{subproof}[1][\proofname]{%
  \begin{proof}[#1]%
}{%
  \end{proof}%
}

\usepackage{verbatim}
\usepackage{algorithm2e}
\usepackage[utf8]{inputenc}
\usepackage{microtype}
\usepackage{amsthm,amsmath,graphicx}
\usepackage{caption}
\usepackage{subcaption}
\usepackage[letterpaper]{hyperref}
\usepackage[table,dvipsnames]{xcolor}
\definecolor{linkblue}{named}{Blue}
\hypersetup{colorlinks=true, linkcolor=linkblue,  anchorcolor=linkblue,
	citecolor=linkblue, filecolor=linkblue, menucolor=linkblue,
	urlcolor=linkblue} 
\usepackage{wasysym}
\usepackage{graphicx}
\usepackage{caption}
\usepackage{enumitem}
\usepackage{thmtools, thm-restate}
\usepackage{wrapfig}
\usepackage{float}
\usepackage{stackrel}
\usepackage{fixfoot}
\usepackage[utf8]{inputenc}
\usepackage[english]{babel}
\usepackage{mathrsfs}  
\usepackage{venturis}
\usepackage{placeins}

\newtheorem{thm}{Theorem}
\newtheorem{lem}[thm]{Lemma}

\newlength\problemsep
\newcommand{\john}[1]{\textcolor{purple}{}}

\newcommand{\thex}{\ensuremath{%
\mathchoice{\includegraphics[page=98, height=2ex]{DirTh6.pdf}}
{\vcenter{\hbox{\includegraphics[page=98, height=2ex]{DirTh6.pdf}}}}
{\includegraphics[page=98, height=1.5ex]{DirTh6.pdf}}
{\includegraphics[page=98, height=1ex]{DirTh6.pdf}}
}}

\newcommand{\hex}{\ensuremath{%
\mathchoice{\includegraphics[page=97, height=2ex]{DirTh6.pdf}}
{\vcenter{\hbox{\includegraphics[page=97, height=2ex]{DirTh6.pdf}}}}
{\includegraphics[page=97, height=1.5ex]{DirTh6.pdf}}
{\includegraphics[page=97, height=1ex]{DirTh6.pdf}}
}}

\title{The Spanning Ratio of the Directed $\Theta_6$-Graph is 5}

\author{Prosenjit Bose, Jean-Lou De Carufel, Darryl Hill, John Stuart}

\begin{document}
\maketitle
\begin{abstract}
    Given a finite set $P\subset\mathbb{R}^2$, the directed Theta-6 graph, denoted $\vec{\Theta}_6(P)$, is a well-studied geometric graph due to its close relationship with the Delaunay triangulation. The $\vec{\Theta}_6(P)$-graph is defined as follows: the plane around each point $u\in P$ is partitioned into $6$ equiangular cones with apex $u$, and in each cone, $u$ is joined to the point whose projection on the bisector of the cone is closest. Equivalently, the $\vec{\Theta}_6(P)$-graph contains an edge from $u$ to $v$ exactly when the interior of $\nabla_u^v$ is disjoint from $P$, where $\nabla_u^v$ is the unique equilateral triangle containing $u$ on a corner, $v$ on the opposite side, and whose sides are parallel to the cone boundaries. It was previously shown that the spanning ratio of the $\vec{\Theta}_6(P)$-graph is between $4$ and $7$ in the worst case (Akitaya, Biniaz, and Bose \emph{Comput. Geom.}, 105-106:101881, 2022). We close this gap by showing a tight spanning ratio of 5. This is the first tight bound proven for the spanning ratio of any $\vec{\Theta}_k(P)$-graph. Our lower bound models a long path by mapping it to a converging series. Our upper bound proof uses techniques novel to the area of spanners. We use linear programming to prove that among several candidate paths, there exists a path satisfying our bound. 
\end{abstract}
\clearpage
\section{Introduction}
\john{Journal to-do: Figures for lemma 8, assumption 9}Studying the distance-preserving properties of particular families of graphs in the plane has been an active area of research for decades. This area has numerous applications, such as motion planning, or wireless routing. We will focus on geometric graphs where each vertex is a point in the plane, and each edge is weighted by the Euclidean distance between its endpoints. Given a geometric graph $G$, the spanning ratio of a subgraph $H$ is a measure of how well distances of $G$ are preserved. In particular, a subgraph $H$ is called a $c$-{\it spanner} of $G$ if for every edge $uv$ in $G$, the shortest path in $H$ from $u$ to $v$ is at most $c$ times the length of the edge $uv$ in $G$ \cite{narasimhan_smid_2007}. The smallest such $c$ is referred to as the {\it spanning ratio} of $H$. Finding a tight bound on the spanning ratio of different geometric graphs is a fundamental problem in computational geometry. A particular graph of interest is the {\it Delaunay triangulation}. It has an edge $uv$ when there exists a disk with $u$ and $v$ on its boundary and no vertices in its interior. The spanning ratio of the Delaunay triangulation has been studied for nearly 40 years, and the first constant upper bound proven was $\frac{1+\sqrt{5}}{2}\pi\approx 5.083$ \cite{DBLP:journals/dcg/DobkinFS90}. Currently, the spanning ratio of the Delaunay triangulation is only known to be between $1.593$ \cite{DBLP:conf/cccg/XiaZ11} and $1.998$ \cite{DBLP:journals/siamcomp/Xia13}, and finding the exact value remains a major open problem \cite{narasimhan_smid_2007}. The first plane geometric graph proven to have a constant spanning ratio is the $\square$-Delaunay graph, whose empty region is a square \cite{DBLP:conf/compgeom/Chew86}. The proof was then generalized to the $\triangle$-Delaunay graph, whose empty region is an equilateral triangle \cite{DBLP:journals/jcss/Chew89}. 

Delaunay triangulations have many desirable properties. A Delaunay triangulation is a plane graph, meaning that it contains roughly $3|V|$ undirected edges, therefore the average vertex degree is approximately $6$. However, it is possible for a vertex to have arbitrarily large degree. In many applications, such as signal transmission, such a large degree may be problematic. Instead, we focus on a closely-related geometric graph with bounded outdegree of $6$.  

Given a finite set $P\subset \mathbb{R}^2$ and integer $k>1$, we can construct a geometric graph on $P$ with at most $k$ outgoing edges per vertex as follows. For each vertex $u\in P$, partition the plane into $k$ equiangular cones about $u$. Then, add an edge from $u$ to the {\it nearest} vertex in each cone about $u$. Repeat this process for each vertex of $P$. If {\it nearest} is measured by the Euclidean distance, then the resulting graph is known as a Yao-$k$ graph, first studied by Yao \cite{DBLP:journals/siamcomp/Yao82}. The directed Theta-$k$ graph, denoted $\vec{\Theta}_k$, is defined similarly, except that {\it nearest} is measured using a distance function whose unit disk is a regular $k$-gon. In particular, the $\vec{\Theta}_k$-graph contains an edge from $u$ to $v$ in cone $C$ when the orthogonal projection of $uv$ onto the cone bisector is minimal among all vertices $v$ in cone $C$. We refer to the undirected version as the $\Theta_k$-graph. Since being introduced by Clarkson \cite{10.1145/28395.28402} and Keil and Gutwin \cite{10.5555/2805869.2805935}, $\Theta_k$-graphs have been studied extensively. For $k\geq 4$, the $\Theta_k$-graph is known to be a spanner of the complete graph \cite{DBLP:conf/wads/BarbaBCRV13,DBLP:journals/dcg/BoseCHS24, DBLP:journals/tcs/BoseCMRV16, 10.1145/28395.28402, 10.5555/2805869.2805935,SeidelRuppert}. On the other hand, the $\vec{\Theta}_k$-graph is known to be a constant spanner for $k\geq 6$ \cite{DBLP:journals/comgeo/AkitayaBB22,DBLP:journals/tcs/BoseCMRV16, SeidelRuppert} and also for $k=4$ \cite{DBLP:journals/dcg/BoseCHS24}. Despite extensive research, tight bounds on the spanning ratio are only known for the $\Theta_{4k+2}$-graph when $k\geq 1$. In general, results claiming tight bounds on the spanning ratio of any geometric graph are rare. To the best of our knowledge, the only other geometric graphs with known tight spanning ratios are the three generalized Delaunay graphs whose empty regions are either a triangle, square, or hexagon. See Table~\ref{tab:compTight}.

\begin{table}[h!]
\centering
\begin{tabular}{||c c c c||} 
 \hline
 Reference & Graph & Spanning Ratio & Notes \\ [0.5ex] 
 \hline\hline
 \cite{DBLP:journals/tcs/BoseCMRV16} & $\Theta_{4k+2}$ ($k\geq 1$) & $1+2\sin(\theta/2)$ & Also for constrained version\cite{DBLP:conf/latin/BoseR14}\\
 \hline
 \cite{DBLP:journals/jcss/Chew89} & $\triangle$-Delaunay & $2$ & Also for affine transformations \cite{bose2025tightroutingspanningratios,DBLP:journals/jgaa/LubiwM19}  \\
 
 \hline
 \cite{DBLP:journals/comgeo/BonichonGHP15} & $\square$-Delaunay & $\sqrt{4+2\sqrt{2}}\approx 2.613$ & Also for affine transformations \cite{DBLP:journals/tcs/BoseCN25,DBLP:conf/esa/RenssenSSW23} \\
 \hline
 \cite{DBLP:journals/jocg/Perkovic0T21} & $\thex$-Delaunay & $2$ & \\[1ex] 
 \hline
\end{tabular}
\caption{Other geometric graphs with known tight spanning ratios. }
\label{tab:compTight}
\end{table}

\john{another table for journal version} The $\vec{\Theta}_6$-graph has a special symmetry. For two points $u,v\in P$, we can define $\nabla_u^v$ to be the unique equilateral triangle whose sides have slopes $-\sqrt{3}$, $0$, and $\sqrt{3}$ and $u$ is on a corner and $v$ is on the opposite side. Then the $\vec{\Theta}_6$-graph contains an edge from $u$ to $v$ exactly when the interior of $\nabla_u^v$ contains no vertices of $P$. Readers familiar with Delaunay triangulations may notice that the $\vec{\Theta}_6$-graph is essentially two rotated copies of a $\triangle$-Delaunay graph. The $\triangle$-Delaunay graph, often referred to as the TD-Delaunay graph, has a spanning ratio of $2$ in the worst case, hence the Theta-$6$ graph inherits this upper bound. On the other hand, the $\vec{\Theta}_6$-graph does not preserve distances nearly as well as its undirected counterpart. The first bounds on the spanning ratio of the $\vec{\Theta}_6$-graph were given by Akitaya, Biniaz, and Bose \cite{DBLP:journals/comgeo/AkitayaBB22}. For any $\epsilon>0$, they construct a $\vec{\Theta}_6$-graph whose spanning ratio is at least $4-\epsilon$. Then, they show that constructing a path in an arbitrary $\vec{\Theta}_6$-graph is not as straightforward as selecting the outgoing edge in the cone containing the destination. This technique is known as {\it greedy routing}, and it produces efficient paths in $\vec{\Theta}_k$-graphs only when $k\geq 7$ \cite{SeidelRuppert}. Greedy routing in $\vec{\Theta}_6$-graphs can lead to a spiraling path of unbounded length. Instead, the authors of \cite{DBLP:journals/comgeo/AkitayaBB22} give an upper bound of $7$ for the spanning ratio of $\vec{\Theta}_6$-graphs using the following technique. Given an edge from $a$ to $b$, they show that the {\it greedy} path from $b$ to $a$ has length at most $6$ times the length of the edge from $a$ to $b$. Then, for any two vertices $s,t$, it suffices to find an intermediate vertex $x$ such that there are short paths from $s$ to $x$ and $t$ to $x$. Indeed, a path from $s$ to $t$ can be formed by concatenating the path from $s$ to $x$ with the greedy paths corresponding to each edge on the path from $t$ to $x$. The authors leave it as an open problem to close the gap between $4$ and $7$. As the title of our paper suggests, we close this gap by showing that the $\vec{\Theta}_6$-graph has a tight spanning ratio of $5$.

Our proof uses elements from \cite{DBLP:journals/comgeo/AkitayaBB22} along with several novel techniques in the area of spanners. Most notably, we use linear programming to prove the existence of a short path. At a high level, our proof is an induction on the distance between points, and it proceeds as follows. We define a special region $R$ between $s$ and $t$, and show that if there exists a vertex $x\in R$, then induction provides two paths, from $s$ to $x$ and from $x$ to $t$, which can be concatenated to form a sufficiently short path from $s$ to $t$. Then, for the rest of the proof, we assume that no vertices are in the interior of $R$. Importantly, we are able to show that any edge $uv$ on the greedy path that crosses $R$ must pay a ``toll''. More precisely, $v$ is guaranteed to be at least twice as close to $t$ as $u$. We leverage this property by applying induction after a candidate path crosses $R$. The difficulty is that there are multiple candidate paths, and determining which path is sufficiently short is non-trivial. Our solution is to assume that there is no sufficiently short candidate path. This allows us to prove inequalities corresponding to each candidate path, then we can use linear programming to show that the system is infeasible. In essence, assuming that there is no short path leads to a contradiction.

\section{Preliminaries}

The geometric graphs we study are weighted, directed graphs whose vertex sets contain points in the plane. We denote the directed edge from vertex $u$ to vertex $v$ by $uv$, and with a slight abuse of notation, when $p,q\in\mathbb{R}^2$, we denote the segment from $p$ to $q$ as $pq$. In general, for an edge $uv$ and a norm $\|\cdot\|$, the value $\|uv\|$ is exactly the factor by which the unit disk of $\|\cdot\|$ must be scaled such that the unique translate with $u$ in the center has $v$ on the boundary. For example, $\|\cdot\|_2$ refers to the Euclidean norm whose unit disk is a circle with radius $1$. We also let $\|uv\|_x$ (resp. $\|uv\|_y$) denote the absolute difference in $x$-coordinate (resp. $y$-coordinate) from $u$ and $v$. The {\it Manhattan norm}, $\|uv\|_1:=\|uv\|_x+\|uv\|_y$, has a tipped square with side length $\sqrt{2}$ as  its unit disk. For an edge $uv$ in a geometric graph, we define its weight to be $\|uv\|_2$, the usual distance from $u$ to $v$. Let $\|\cdot\|_{\hex}$ denote the {\it hexagonal norm} whose unit disk is a regular hexagon with east and west vertices equal to $(1,0)$ and $(-1,0)$, respectively. On the other hand, let $\|\cdot\|_{\thex}$ denote the {\it tipped hexagonal norm} whose unit disk is a tipped regular hexagon with north and south vertices equal to $(0,\frac{2}{\sqrt{3}})$ and $(0,\frac{-2}{\sqrt{3}})$, respectively. See Figure~\ref{fig:normsAndCones}. Notice that we always have $\|uv\|_{\thex}\leq\|uv\|_2\leq \|uv\|_{\hex}$, and equality holds exactly when $uv$ has slope $-\sqrt{3}$, $0$, or $\sqrt{3}$. 
\begin{figure}[ht!]
        \centering
        \includegraphics[page=79, scale=0.9]{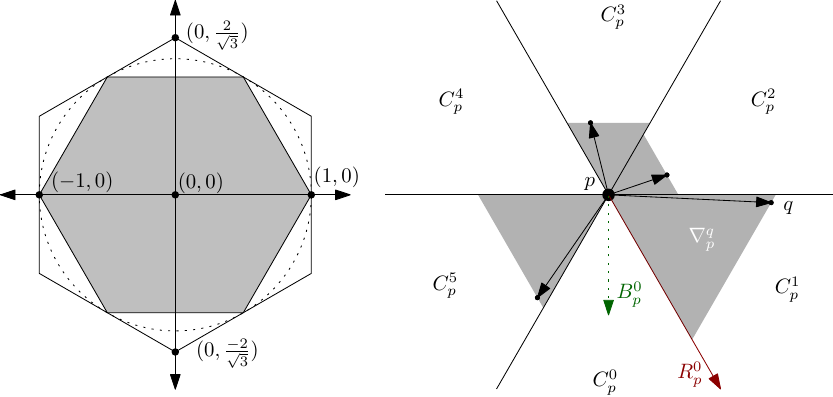}
        \caption{Left: The dotted circle represents the unit disk in the standard Euclidean norm ($\|\cdot\|_2$). The shaded inner hexagon is the unit disk in the $\|\cdot\|_{\hex}$-norm. The outer tipped hexagon is the unit disk in the $\|\cdot\|_{\thex}$-norm. Right: The six cones about $p$ are labeled in counterclockwise order from $ C_p^0$ below $p$. The dotted green bisector is $B_p^0$ and the red ray is $R_p^0$. The empty shaded triangle in $ C_p^1$ is $\nabla_p^q$. Each shaded region is empty. \label{fig:normsAndCones}}
    \end{figure}

In a geometric graph $G$, we define the length of a path to be the sum of the weights of its edges. For any norm $\|\cdot\|$ and path $\mathcal{P}$, we define $\|\mathcal{P}\|:=\sum_{uv\text{ on }\mathcal{P}}\|uv\|$. If $v$ is a vertex on a path $\mathcal{P}$, then we define the {\it truncated path} $\mathcal{P}_v$ as the subpath of $\mathcal{P}$ which ends at $v$. For two vertices $u,v$ in a geometric graph $G$, we let $d_G(u,v)$ denote the length of the shortest path in $G$ from $u$ to $v$. Recall that all edges are weighted by the Euclidean norm. Then for a constant $c\geq 1$, $G$ is said to be a $c$-spanner if for all points $u,v$ in $G$, we have $d_G(u,v)\leq c\|uv\|_2$. The {\it spanning ratio} of $G$ is the least $c$ for which $G$ is a $c$-spanner. The spanning ratio of a class of graphs $\mathcal{G}$ is the least $c$ for which all graphs in $\mathcal{G}$ are $c$-spanners.

For any point $p$, consider the three lines passing through $p$ with slopes $-\sqrt{3}$, $0$, and $\sqrt{3}$. These lines partition the plane into six cones about $p$, which we label counter-clockwise by $ C_p^i$ for $i\in\mathbb{Z}_6$ starting with $ C_p^0$ directly below $p$. Here, $\mathbb{Z}_6$ refers to the integers modulo $6$. See Figure~\ref{fig:normsAndCones}. For each cone $ C_p^i$, we denote its bisector by $B_p^i$. Throughout this paper, we adopt the convention that cones are closed sets, otherwise we write $\text{int}(R)$ to denote the interior of region $R$. For $i\in\mathbb{Z}_6$, we define the ray $R_p^i:= C_p^i\cap C_p^{i+1}$. For points $p,q\in\mathbb{R}^2$ with $q\in C_p^i$, we let $\nabla_p^q$ denote the triangular region $\{r\in  C_p^i\mid \|pr\|_{\hex}\leq\|pq\|_{\hex}\}$. For any equilateral triangle $T$, we let $\|T\|$ denote the side length of $T$. Notice that $\|\nabla_p^q\|=\|pq\|_{\hex}$. For any finite point set $P\subset\mathbb{R}^2$, we make the {\em general position assumption} that no two vertices lie on a line with slope $-\sqrt{3}$, $0$, or $\sqrt{3}$. The directed Theta-6 graph of $P$, denoted $\vec{\Theta}_6(P)$, contains an edge from vertex $u$ to vertex $v$ exactly when $\text{int}(\nabla_u^v)\cap P=\emptyset$. Intuitively, a vertex has an outgoing edge to the nearest vertex in each of its six cones, measured using the $\|\cdot\|_{\hex}$-norm. For any vertices $u,v\in P$, we let $\pi(u,v)$ denote the {\it greedy path} from $u$ to $v$. In particular, every edge $ab$ of $\pi(u,v)$ satisfies $b\in C_a^i\iff v\in C_a^i$. If $w$ is a vertex along the path $\pi(u,v)$, then the subpath terminating at $w$ is denoted by $\pi(u,v)_w$. Note that if $w\neq v$, then $\pi(u,v)_w$ is not necessarily the same as $\pi(u,w)$.

\section{Lower Bound}
\begin{lem}\label{lem:SRLBDT6}
For any $\epsilon>0$, there exists a directed Theta-$6$-graph with spanning ratio at least $5-\epsilon$.
\end{lem}
\begin{proof}
    Let $\epsilon>0$ be arbitrarily small. We will construct a $\vec{\Theta}_6$-graph with spanning ratio at least $5-\epsilon$, shown in Figure~\ref{fig:DirT6SRLB}.
    \begin{figure}[ht!]
        \centering
        \includegraphics[page=99, scale=0.8]{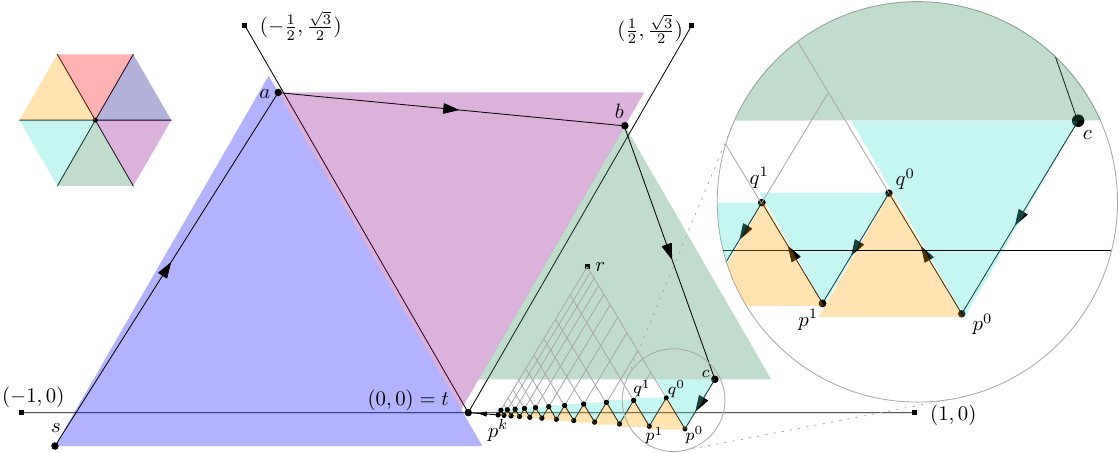}
        \caption{Graph construction achieving a spanning ratio lower bound of $5-\epsilon$. The shortest path from $s$ to $t$ is the greedy path $s,a,b,c,p^0,q^0,p^1,q^1,...,,q^{k-1},p^k,t$. Top Left: The colour coding for edges based on direction. \label{fig:DirT6SRLB}}
    \end{figure}
    For the following construction, set $0<\delta<\frac{\epsilon}{72}$. Define the points $t:=(0,0)$, $s:=(-1,0)+\delta(1,-1),a:=(-\frac{1}{2},\frac{\sqrt{3}}{2})+\delta(1,-2),b:=(\frac{1}{2},\frac{\sqrt{3}}{2})+\delta(-2,-3), c:=(1,0)+\delta(-6,1)$. Next, let $p^0$ be the unique point such that $tp^0$ has slope $-\delta$ and $p^0c$ has slope $\sqrt{3}-\delta$. Next, let $q^0$ be on segment $ct$ such that $p^0q^0$ has slope $\delta-\sqrt{3}$. For $i>0$, let $q^i$ be on $ct$ such that $\frac{\|q^{i-1}q^i\|_2}{\|tq^{i}\|_2}=\frac{\|q^0c\|_2}{\|tc\|_2}$. For $i>0$, also define $p^i$ to be on $tp^0$ such that $p^iq^i$ has slope $\delta-\sqrt{3}$. We let $k$ be the least integer such that $p_k$ has $x$-coordinate less than $\delta$. We have constructed the point set $P:=\{t,s,a,b,c\}\cup\{p^0,p^1,...,p^k\}\cup \{q^0,q^1,...,q^{k-1}\}$. In $\vec{\Theta}_6(P)$, the shortest path from $s$ to $t$ is $s,a,b,c,p^0,q^0,p^1,q^1,...q^{k-1},p^k,t$ since from the perspective of any vertex on the path, all subsequent vertices lie in the same cone. Now we will lower bound the length of this path. First, we have  
    \begin{align}\label{eq:sa}
        1 &= \|(-1,0)(-\frac{1}{2},\frac{\sqrt{3}}{2})\|_2 &\text{since $\sqrt{\left(\frac{1}{2}\right)^2+\left(\frac{\sqrt{3}}{2}\right)^2}=1$,}\nonumber\\
        &\leq \|(-1,0)s\|_2 + \|sa\|_2 + \|a(-\frac{1}{2},\frac{\sqrt{3}}{2})\|_2 &\text{by triangle inequality,}\nonumber\\
        &\leq \|(-1,0)s\|_1 + \|sa\|_2 + \|a(-\frac{1}{2},\frac{\sqrt{3}}{2})\|_1&\text{since $\|\cdot\|_2\leq \|\cdot\|_1$,}\nonumber\\
        &= 2\delta + \|sa\|_2 + 3\delta &\text{by definition.}
    \end{align}
    which implies $\|sa\|_2\geq 1-5\delta$. Using the same technique, we can obtain the following bounds: $\|ab\|_2\geq 1-8\delta$ and $\|bc\|_2\geq 1-12\delta$. Next, let $r$ be the unique point such that $p^k r$ has slope $\sqrt{3}-\delta$ and $rp^0$ has slope $\delta-\sqrt{3}$. Then by unfolding, notice that the path $p^0,q^0,p^1,q^1,...,p^{k-1},q^{k-1},p^k$ has length $\|p^0r\|_2+\|rp^k\|_2$. See Figure~\ref{fig:DirT6SRLB}. By the triangle inequality, we have $\|p^0(1,0)\|_1\leq \|p^0c\|_x+\|p^0c\|_y+\|c(1,0)\|_1\leq 2\delta+2\delta+7\delta=11\delta$ and $\|tp^k\|_1= \|tp^k\|_x + \|tp^k\|_y\leq \delta+\delta=2\delta$. Next, we claim that $r$ approaches $(\frac{1}{2},\frac{\sqrt{3}}{2})$ as $\delta\to 0$.
    \begin{clm}\label{claim:lb}
        We have $\|r(\frac{1}{2},\frac{\sqrt{3}}{2})\|_2\leq 17\delta$.
    \end{clm}
    \begin{subproof}
        Indeed, first let $r^1$ be the unique point such that $p^kr^1$ has slope $\sqrt{3}$ and $r^1p^0$ has slope $-\sqrt{3}$. See Figure~\ref{fig:rNearTop}. Since $\max(\|p^kr^1\|_x,\|r^1p^0\|_x)\leq 1$ vertical distance from $r$ to the lines through $p^kr^1$ and $r^1p^0$ is at most $\delta$. Then $r^1$ must be in the ball of radius $\delta$ in the $\|\cdot\|_1$-norm centered at $r$, meaning $\|rr^1\|_1\leq\delta$. Next, let $r^2$ be the unique point such that $p^kr^2$ has slope $\sqrt{3}$ and $r^2(1,0)$ has slope $-\sqrt{3}$. Then we have $\|r^1r^2\|_2\leq \frac{2}{\sqrt{3}}\|p^0(1,0)\|_2$. Similarly $\|r^2(\frac{1}{2},\frac{\sqrt{3}}{2})\|_2\leq \frac{2}{\sqrt{3}}\|tp^k\|_2$. Putting this together, we obtain 
        \begin{align*}
            \|r(\frac{1}{2},\frac{\sqrt{3}}{2})\|_2\leq \|rr^1\|_1+\|r^1r^2\|_2+\|r_2(\frac{1}{2},\frac{\sqrt{3}}{2})\|_2\leq \delta+\frac{2}{\sqrt{3}}(11\delta)+ \frac{2}{\sqrt{3}}(2\delta)\leq  17\delta.
        \end{align*}
        \begin{figure}[ht!]
            \centering
            \includegraphics[page=96, scale=0.8]{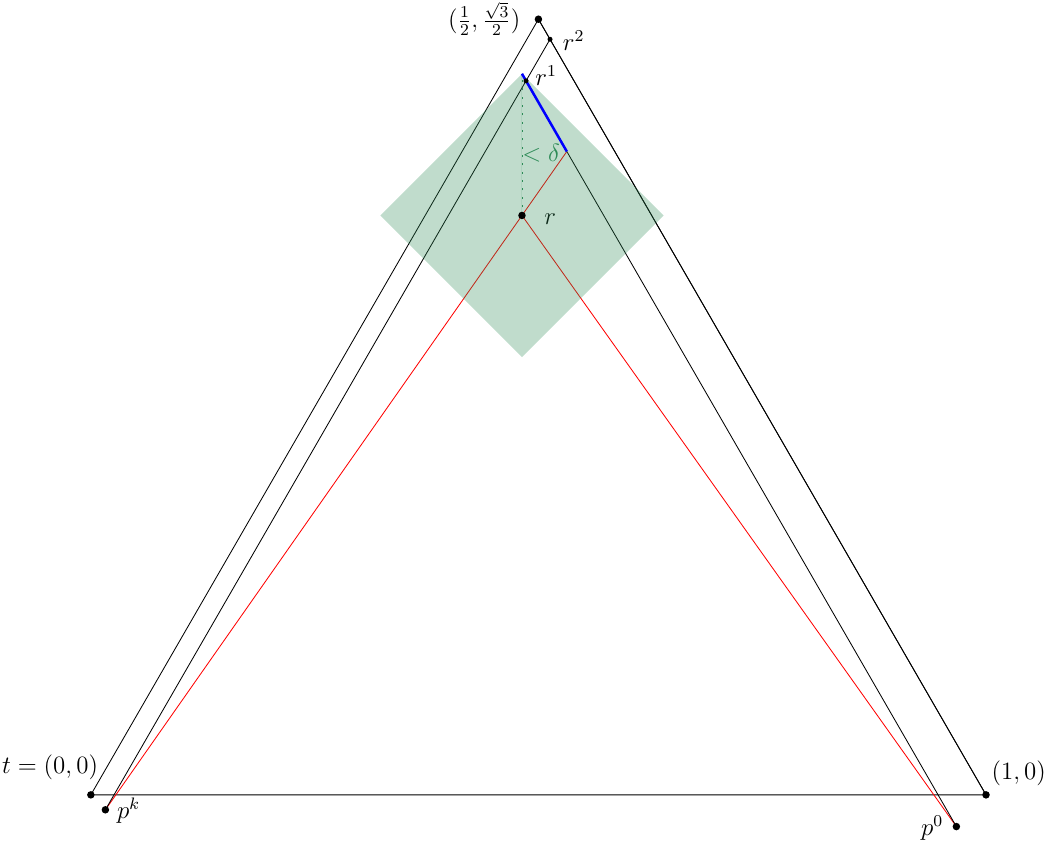}
            \caption{In Claim~\ref{claim:lb}, we define $r^1$ and $r^2$ in order to show that $r$ approaches $(\frac{1}{2},\frac{\sqrt{3}}{2})$. The red lines have slope $\pm (\sqrt{3}-\delta)$, whereas the black lines have slope $\pm\sqrt{3}$. In this example, the line through $r^1 p^0$ has a larger vertical distance to $r$ than the line through $p^k r^1$. This distance is at most $\delta$. Since $r^1$ is to the right of $r$ in this example, $r^1$ must lie on the blue segment, which is inside the green disk centered at $r$ with radius at most $\delta$ in the $\|\cdot\|_1$-norm. If instead $r^1$ were to the left of $r$, then the diagram would be similar.} \label{fig:rNearTop}
        \end{figure}
    \end{subproof}
    Now we can apply the same technique used to prove inequality~\eqref{eq:sa} to obtain $\|p^0r\|_{2}\geq 1-28\delta$ and $\|rp^k\|_2\geq 1-19\delta$. Therefore the total length of the path is at least 
    \begin{align*}
        d_{\vec{\Theta}_6(P)}(s,t) &=\|sa\|_2+\|ab\|_2+\|bc\|_2+\|cp^0\|_2 + \left(\sum_{i=0}^{k-1}\|p^iq^i\|_2\right)+\|q^{k-1}p^k\|_2+\|p_k t\|_2\\
        &\geq (1-5\delta)+(1-8\delta)+(1-12\delta) + (1-28\delta)+(1-19\delta)\\
        &=5-72\delta>5-\epsilon.
    \end{align*}
    Then since $\|st\|_2\leq 1$, the spanning ratio of $\vec{\Theta}_6(P)$ is at least $5-\epsilon$.
\end{proof}

\section{Upper Bound}\label{sec:UpperBound}

The goal of this section is to prove Theorem~\ref{thm:sr}.
\begin{theorem}\label{thm:sr}
    For any point set $P\subset\mathbb{R}^2$ in general position, we have $d_{\vec{\Theta}_6(P)}(s,t)\leq 5\|st\|_{\thex}\leq 5\|st\|_2$ $\forall s,t\in P$.
\end{theorem}
Let $P\subset\mathbb{R}^2$ be a point set in general position, and let $G:=\vec{\Theta}_6(P)$. Recall that the $\|\cdot\|_{\hex}$-norm is used to define $G$, however Theorem~\ref{thm:sr} is expressed in terms of the $\|\cdot\|_{\thex}$-norm. The following lemma describes the exact relationship between $\|\cdot\|_{\hex}$ and $\|\cdot\|_{\thex}$.

\begin{lem}\label{lem:hexthex}
    Let $u\in\mathbb{R}^2$, $i\in\mathbb{Z}_6$, and $v\in C_u^i$. Then $$\|uv\|_{\hex}=\|uv\|_{\thex}+\frac{1}{2}\min(\|C_u^i\cap C_v^{i-2}\|,\|C_u^i\cap C_v^{i+2}\|).$$
\end{lem}
\begin{proof}
    Suppose without loss of generality that $i=0$ and $v$ is to the right of the bisector $B_u^0$. See Figure~\ref{fig:hexthex}. Then $\|C_u^0\cap C_v^{2}\|\leq \|C_u^0\cap C_v^{4}\|$. Let $a$ be the orthogonal projection of $v$ on $R_u^0$. Also let $b$ be the intersection of $R_u^0$ with a horizontal line through $v$. Then 
    $$\|uv\|_{\thex}+\frac{1}{2}\|C_u^0\cap C_v^{2}\|=\|ua\|_{2}+\|ab\|_{2}=\|ub\|_2=\|uv\|_{\hex}.$$
    The other cases follow by symmetry.
    \begin{figure}[ht!]
        \centering
        \includegraphics[page=80, scale=0.9]{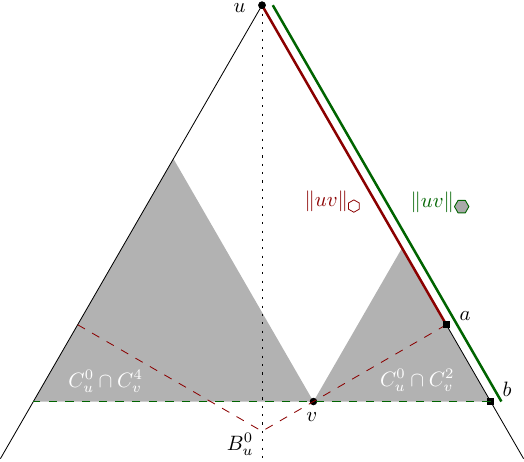}
        \caption{Lemma~\ref{lem:hexthex} with $v\in C_u^0$. The dotted vertical ray is the bisector $B_u^0$. The horizontal green dashed segment represents all points $p\in C_u^0$ such that $\|up\|_{\hex}=\|uv\|_{\hex}$. The red dashed segments represent all points $p\in C_u^0$ such that $\|up\|_{\thex}=\|uv\|_{\thex}$.} \label{fig:hexthex}
    \end{figure}
\end{proof}
The next lemma finds use in proving Theorem~\ref{thm:sr}
\begin{lem}\label{lem:WithinR}
    Let $u,v,t\in P$ and $i\in\mathbb{Z}_6$. Suppose $u,v\in C_t^i\cup C_t^{i+1}$ between $B_t^i$ and $B_t^{i+1}$. Then if $v\in  C_u^{i+3}\cup C_u^{i+4}$, we have $\|uv\|_{\hex}\leq 2(\|ut\|_{\thex}-\|vt\|_{\thex})$.
\end{lem}
\begin{proof}
    Without loss of generality, suppose $i=0$. If $v\in C_u^3$, then define $p:=R_u^{2}\cap R_v^1$, otherwise define $p:=R_u^{2}\cap R_v^1$. Let $v'$ and $p'$ be the orthogonal projections of $v$ and $p$ on $R_u^{3}$. See Figure~\ref{fig:WithinR}. Then $$\|uv\|_{\hex}=\|up\|_{2}=2\|up'\|_{2}\leq 2\|uv'\|_2=2 (\|ut\|_{\thex}-\|vt\|_{\thex}).$$
    \begin{figure}[ht!]
        \centering
        \includegraphics[page=82, scale=0.9]{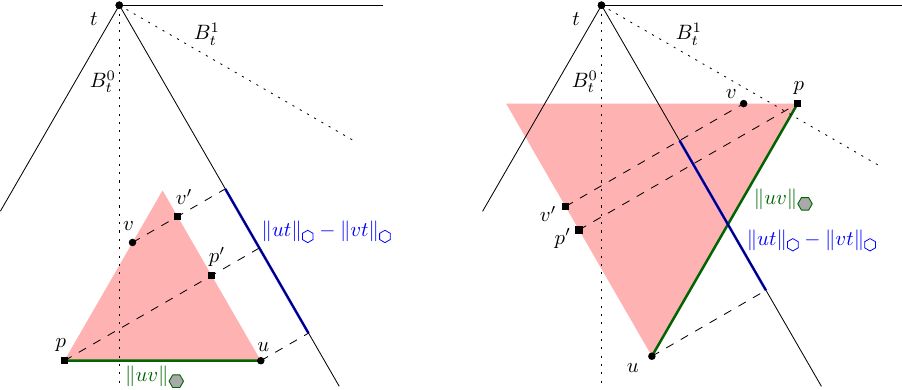}
        \caption{For the proof of Lemma~\ref{lem:WithinR}, we assume $i=0$ and show the case when $v\in C_u^4$ on the left and $v\in C_u^3$ on the right. \label{fig:WithinR}}
    \end{figure}
\end{proof}

We prove Theorem~\ref{thm:sr} by induction on the hexagonal distance $\|\cdot\|_{\hex}$ between vertices. For the base case, consider the closest pair of vertices $s,t\in P$ (i.e. $\|st\|_{\hex}$ is minimal). Minimality implies $\text{int}(\nabla_s^t)$ must not contain any vertices of $P$. Therefore there must be an edge from $s$ to $t$, yielding $d_{G}(s,t)=\|st\|_2\leq 5\|st\|_{\thex}$. 

For the remainder of the proof, fix $s,t\in P$ and suppose that for any vertices $u,v\in P$, if $\|uv\|_{\hex}<\|st\|_{\hex}$, then $d_{G}(u,v)\leq 5\|uv\|_{\thex}$.  We will assume $\|st\|_{\hex}=1$ since our argument holds when scaling the point set. Without loss of generality, suppose $s\in C_t^0$, with $s$ being to the right of $B_t^0$. Define the triangle $T=C_t^0\cap C_s^2$, and denote its side length by $y_0:=\|T\|$. Notice that $y_0$ describes the {\it hex} distance from $s$ to the nearest boundary, and by Lemma~\ref{lem:hexthex}, we obtain $5\|st\|_{\thex}=5-2.5y_0$. Suppose towards a contradiction that $5\|st\|_{\thex}<d_G(s,t)$. Our approach is to analyse paths from $s$ to $t$, which will allow us to prove inequalities based on our faulty assumption. These inequalities will form a linear system, then linear programming will help us show that the system is infeasible. The first step in our proof is to establish an empty region which we refer to as $R$. 
\begin{definition}[Empty Region $R$]
    Consider the points $p:= B_s^3\cap B_t^1$ and $q=B_s^4\cap B_t^0$. See Figure~\ref{fig:SideBySide}. Let $Q$ denote the parallelogram with vertices $s,q,t,p$, and define the region $R:=Q\cup \nabla_s^p\cup \nabla_s^q$.
\end{definition}

\begin{figure}[ht!]
    \centering
    \includegraphics[page=89, scale=0.9]{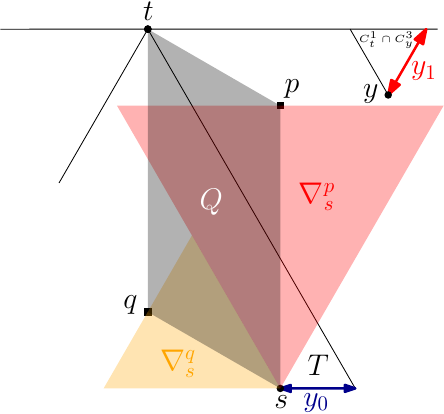}
    \caption{The red triangle is $\nabla_s^p$, the orange triangle is $\nabla_s^q$, and the shaded parallelogram is $Q$. Their union is $R$, whose interior can be assumed to not contain any vertices by Lemma~\ref{lem:empty}. Furthermore, there exists a vertex $y\in C_t^1$ close to the horizontal line through $t$. The triangle $T$ has side-length $y_0$ and the triangle $C_t^1\cap C_y^3$ has side-length $y_1$. \label{fig:SideBySide}}
\end{figure}

Let $y$ be the vertex such that $sy$ is the edge in $C_s^3$ and consider the quantity $y_1:=1-\|\nabla_s^y\|$ describing the hexagonal distance $\|\cdot\|_{\hex}$ from $y$ to the horizontal line through $t$. In addition to proving that $R$ is empty, the following lemma proves that $y_0$ and $y_1$ cannot both be large. Geometrically, $s$ must be near $R_t^0$ and $y$ must be near $R_t^1$. See Figure~\ref{fig:SideBySide}.

\begin{lem}\label{lem:empty}
    If $d_G(s,t)>5\|st\|_{\thex}$, then $\text{int}(R)$ does not contain any vertices of $P$. Furthermore, $y\in C_t^1$ and
    \begin{align}
        2.5y_0+3.5y_1<1 \label{eq:firstAssumption}
    \end{align}
\end{lem}
\begin{proof} 
    See Figure~\ref{fig:EmptyRegion}. If there exists a vertex $w\in Q$ with $w\notin\{s,t\}$, then $\|wt\|_{\hex}<1$. We obtain
    \begin{align*}
        d_{G}(s,t) &\leq d_{G}(s,w)+d_{G}(w,t) &\text{by triangle inequality,}\\
        &\leq 5\|sw\|_{\thex}+5\|wt\|_{\thex} &\text{by induction,}\\
        &=5\|st\|_{\thex} &\text{by $\|\cdot\|_{\thex}$ collinearity.}
    \end{align*}
    \begin{figure}[ht!]
        \centering
        \includegraphics[page=67, scale=1]{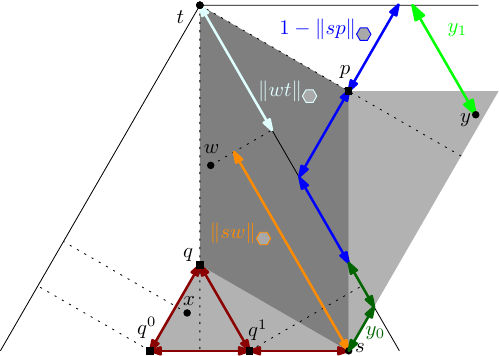}
        \caption{If there is a vertex $w$ in the dark shaded parallelogram, denoted $Q$, then induction yields paths from $s$ to $w$ and from $w$ to $t$. If there is an edge $sy$ in the light shaded triangle $\nabla_s^p\setminus Q$, then induction yields a path from $y$ to $t$. Similarly, if there is an edge $sx$ in the light shaded triangle $\nabla_s^q\setminus Q$, then induction yields a path from $x$ to $t$. \label{fig:EmptyRegion}}
    \end{figure}
    Therefore we may assume parallelogram $Q$ contains no vertices. Next we will rule out the possibility of $y\in C_t^5\cup C_t^0$. If $y\in C_t^5$, then we have 
    \begin{align*}
        d_G(s,t) &\leq \|sy\|_{\hex}+d_G(y,t) &\text{by triangle inequality,}\\
        &\leq 1+5y_0 &\text{by induction as $\|yt\|_{\thex}\leq\|yt\|_{\hex}=\|\nabla_t^y\|\leq\|T\|$}\\
        &< 5\|st\|_{\hex} &\text{since $y_0\leq 0.5$ and $5-2.5y_0=5\|st\|_{\hex}$.}
    \end{align*}
    On the other hand if $y\in  C_t^0$ or $y\in C_t^1$ below $B_t^1$, then we use Lemma~\ref{lem:WithinR} to obtain $\|sy\|_{\hex}\leq 2(\|st\|_{\thex}-\|yt\|_{\thex})$. Indeed, if $sy$ crosses $B_t^0$, then let $z:=sy\cap B_t^0$ and apply Lemma~\ref{lem:WithinR} to both $sz$ and $zy$. We obtain
    \begin{align*}
        d_G(s,t) &\leq \|sy\|_{\hex}+d_G(y,t) &\text{by triangle inequality,}\\
        &\leq 2(\|st\|_{\thex}-\|yt\|_{\thex}) + 5\|yt\|_{\thex} &\text{by Lemma~\ref{lem:WithinR} and induction,}\\
        &= 2\|st\|_{\thex} + 3\|yt\|_{\thex}\\
        &<5\|st\|_{\thex} &\text{since $\|yt\|_{\thex}<\|st\|_{\thex}$.}
    \end{align*}
    Therefore we may assume $y\in C_t^1$ above the bisector $B_t^1$. We have $\|yt\|_{\hex}\leq\|C_t^1\cap C_s^3\|=1-y_0$. Recall that $y_1:=1-\|\nabla_s^y\|$, and notice that we also have $y_1=\|C_t^1\cap C_y^{3}\|$. By Lemma~\ref{lem:hexthex}, $\|yt\|_{\thex}=\|yt\|_{\hex}-\frac{1}{2}y_1$. This gives us
    \begin{align}
        5-2.5y_0 &<d_{G}(s,t) &\text{by assumption,} \nonumber\\
        &\leq \|sy\|_{\hex}+d_{G}(y,t) &\text{by triangle inequality,} \nonumber\\
        &\leq (1-y_1)+ (5(1-y_0)-2.5y_1) &\text{by $\|\nabla_s^y\|=\|sy\|_{\hex}$ and induction.} \label{eq:AlmostfirstAssumption}
    \end{align}
    By re-arranging inequality \eqref{eq:AlmostfirstAssumption}, we obtain $2.5y_0+3.5y_1<1$, proving \eqref{eq:firstAssumption}. Notice that $2y_0+3(1-\|sp\|_{\hex})=1$ by symmetry of isosceles triangles. We have
    \begin{align*}
        y_1< \frac{1-2.5y_0}{3.5}<\frac{1-2y_0}{3}=1-\|sp\|_{\hex},
    \end{align*}
    which implies that the region $\text{int}(\nabla_s^p)$ contains no vertices. Indeed, if there were a point $u\in \text{int}(\nabla_s^p)$, then $1-y_1=\|sy\|_{\hex}\leq \|su\|_{\hex}\leq \|sp\|_{\hex}$. Thus far we have shown that $Q$ and $\nabla_s^p$ do not contain vertices in their interior, hence it remains to argue that $\text{int}(\nabla_s^q)$ is also empty. Towards a contradiction, assume there is a vertex $x\in \text{int}(\nabla_s^q)$. If $x$ is to the right of the bisector $B_t^0$, then we have
    \begin{align*}
        d_{G}(s,t)&\leq \|sx\|_{\hex}+d_{G}(x,t) &\text{by triangle inequality,}\\
        &\leq 2(\|st\|_{\thex}-\|xt\|_{\thex})+ 5\|xt\|_{\thex} &\text{by Lemma~\ref{lem:WithinR} and induction,} \\
        &=2\|st\|_{\thex}+3\|xt\|_{\thex}\\
        &< 5\|st\|_{\thex} &\text{since $\|xt\|_{\thex}<\|st\|_{\thex}$.}
    \end{align*}
    On the other hand if $x$ is to the left of the bisector $B_t^0$, then let $q^0:= R_q^5\cap R_s^4$ and $q^1:= R_q^0\cap R_s^4$. Notice that $\|sq^0\|_{\hex}=2\|sq^1\|_{\hex}$. We have 
    \begin{align*}
        d_{G}(s,t)&\leq \|sx\|_{\hex}+d_{G}(x,t) &\text{by triangle inequality,}\\
        &\leq \|sq^0\|_{\hex}+ 5\|xt\|_{\thex} &\text{by induction,}\\
        &\leq 2\|sq^1\|_{\hex}+ 5\|q^0t\|_{\thex} &\text{since $\|xt\|_{\thex}\leq\|q^0t\|_{\thex}$,}\\
        &=2(2(\|st\|_{\thex}-\|q^1t\|_{\thex}))+5\|q^1t\|_{\thex} &\text{by Lemma~\ref{lem:WithinR} and $\|q^0t\|_{\thex}=\|q^1t\|_{\thex}$,}\\
        &= 4\|st\|_{\thex} +  \|q^1t\|_{\thex}\\
        &\leq 5\|st\|_{\thex} &\text{since $\|q^1t\|_{\thex}\leq\|st\|_{\thex}$.}
    \end{align*}
    Therefore the region $\text{int}(\nabla_s^q)$ contains no vertices. This completes the proof. 
\end{proof}

Now that we have established an empty region $R$ and the position of $y$, we would like to analyse the greedy path from $y$ to $t$. Recall that for any vertices $u,v\in P$, the greedy path from $u$ to $v$ is denoted $\pi(u,v)$. Every edge $ab$ of $\pi(u,v)$ satisfies $b\in C_a^i\iff v\in C_a^i$. Intuitively, the greedy path simply follows the edge in the direction of the destination. In \cite{DBLP:journals/comgeo/AkitayaBB22} the authors show that for every edge $uv$ of $\pi(s,t)$, we have $\|ut\|_{\hex}>\|vt\|_{\hex}$. However, this does not guarantee that the greedy path is an efficient path to $t$. The main reason is that $\pi(s,t)$ may spiral around $t$ arbitrarily many times while making negligible progress. The authors of \cite{DBLP:journals/comgeo/AkitayaBB22} show how to avoid a spiraling greedy path by assuming there is an edge from $t$ to $s$ and using the empty triangle $\nabla_t^s$ as an obstruction. This is done in Lemmas 2, 3 and 4 from \cite{DBLP:journals/comgeo/AkitayaBB22}, which we restate here using our notation.
\begin{lem}[Lemma 2 in \cite{DBLP:journals/comgeo/AkitayaBB22}]\label{lem:adj}
    If edge $uv$ is on $\pi(u,t)$ and $u\in  C_t^i$, then $v\in C_t^{i-1}\cup  C_t^i\cup  C_t^{i+1}$.
\end{lem}

\begin{lem}[Lemma 3 in \cite{DBLP:journals/comgeo/AkitayaBB22}]\label{lem:3from1}
    Let $s,t\in P$ and assume $ts$ is an edge of $\vec{\Theta}_6(P)$. Assume the vertices $a,b,c$ appear in this order in $\pi(s,t)$, not necessarily distinct or consecutive. Then if $a$ and $c$ are both in the same cone $C_t^i$ for $i\in\mathbb{Z}_6$, then $\text{int}(\nabla_a^b)\cap\text{int}(\nabla_t^c)=\emptyset$. 
\end{lem}

\begin{lem}[Lemma 4 in \cite{DBLP:journals/comgeo/AkitayaBB22}]
    Let $s,t\in P$ and assume $ts$ is an edge of $\vec{\Theta}_6(P)$. Then for $i\in\mathbb{Z}_6$, we have
    $$\sum\limits_{\substack{ab\in \pi(s,t)\\ a\in C_t^i}} \|ab\|\leq \|\nabla_t^s\|$$
\end{lem}

Instead of working with an empty triangle that forces the greedy path not to spiral as in \cite{DBLP:journals/comgeo/AkitayaBB22}, we will ignore any spiraling by only using the greedy path up to and including the first edge that crosses the boundary $R_t^0$. Despite this change, the arguments in \cite{DBLP:journals/comgeo/AkitayaBB22} can be applied in our setting to prove Lemma~\ref{lem:greedyPotential} which is a slightly modified version of Lemma 4 from \cite{DBLP:journals/comgeo/AkitayaBB22}. 

\begin{lem}\label{lem:greedyPotential}
    Fix $j\in\mathbb{Z}_6$ and let $s,t,u,v\in P$ such that $uv$ is an edge of $\pi(s,t)$ and no edge of the truncated greedy path $\pi(s,t)_u$ crosses $R_t^j$. For $i\in\mathbb{Z}_6$, let $f_i$ denote the first vertex of $\pi(s,t)_{u}$ in $C_t^i$. If no such vertex $f_i$ exists, let $f_i:=t$. Then $\|\pi(s,t)_{v}\|_{\hex}\leq\sum_{i=0}^5 \|f_it\|_{\hex}$.
\end{lem}
\begin{proof}
    Suppose without loss of generality that $i=0$. For any point $p$ below $t$, let $p'$ denote the intersection of the horizontal line through $p$ with the ray $R_t^0$. We claim that if $ab$ and $cd$ are edges of $\pi(s,t)_{v}$ such that $a,c\in C_t^0$, then the segments $a'b'$ and $c'd'$ are disjoint. Furthermore, both $a'b'$ and $c'd'$ are subsegments of $f_0't$. This claim follows from Lemma~\ref{lem:3from1} of \cite{DBLP:journals/comgeo/AkitayaBB22} since the path $\pi(s,t)_u$ is assumed to not cross $R_t^j$. We sum all edges with source in $C_t^0$ to obtain
    \begin{align*}
        \sum_{\substack{ab\text{ on }\pi(s,t)_{v}\\ a\in C_t^0}} \|ab\|_{\hex}= \sum_{\substack{ab\text{ on }\pi(s,t)_{v}\\ a\in C_t^0}} \|a'b'\|_{\hex}\leq \|f_0't\|_{\hex}=\|f_0t\|_{\hex}.
    \end{align*}
    The lemma follows by summing over the edges in all cones $C_t^i$ for $i\in\mathbb{Z}_6$.
\end{proof} 
For the rest of the section, we make Assumptions~\ref{asm:q} and~\ref{asm:x} to help us prove Theorem~\ref{thm:sr}. We then consider the case where Assumption~\ref{asm:q} is not satisfied in Section~\ref{sec:notq}. On the other hand, Section~\ref{sec:onlyq} treats the case when Assumption~\ref{asm:x} does not hold. The first assumption roughly states that the greedy path from $y$ spirals counterclockwise around $t$ without making significant progress to $t$. We define $u^y v^y$ to be the first edge of $\pi(y,t)$ which crosses $R_t^0$. If such an edge does not exist, let $u^y:=t$ and $v^y:=t$. Then the first assumption can be stated as follows:
\begin{assumption}\label{asm:q}
    There exists a vertex $q\in C_t^5$ on $\pi(y,t)_{v^y}$ with $\|qt\|_{\hex}>y_0$. 
\end{assumption}
Our second assumption is that there is no edge $uv$ crossing $R_t^5$ such that $u\in T$ and $v\in C_u^4$.
\begin{assumption}\label{asm:x}
    For every vertex $u\in T$, we have  $P\cap\text{int}( C_t^0\cap C_u^4)\neq\emptyset$. 
\end{assumption}
Let $x$ be the vertex such that $sx$ is the edge in $C_s^4$. By Assumption~\ref{asm:x}, $x$ exists and $x\in C_t^0$. Combining our two assumptions allows us to naturally consider two candidate paths, loosely referred to as the {\it $x$-path} and the {\it $y$-path}, and show that one of them is sufficiently short to force a contradiction with $d_G(s,t)>5\|st\|_{\thex}$. Informally, the $x$-path travels clockwise around $t$, whereas the $y$-path travels in the opposite direction. Where these two paths meet, they must compete for space. The lengths of the paths are primarily dictated by their {\it long} edges which cross cone boundaries around $t$. These {\it long} edges come with large empty regions which squeeze the other path closer to $t$. We formalize this intuition below.
    
Recall that $u^y v^y$ is the first edge of $\pi(y,t)$ that crosses the boundary $R_t^0$. Consider the candidate path $sy$ concatenated with $\pi(y,t)_{v^y}$, which we loosely refer to as the $y$-path, ending at $v^y$. Overall, the $y$-path progresses counterclockwise about $t$ until $C_t^5$, however it is still possible that the $y$-path reverses direction and ends up crossing $R_t^0$ from $C_t^1$ to $C_t^{0}$. Maintaining full generality, we may only assume that for some $i\in\{0,1\}$, we have $u^y\in C_t^i$ and $v^y\in C_t^{1-i}$, which follows from Lemma~\ref{lem:adj}. In light of Lemma~\ref{lem:greedyPotential}, we need not concern ourselves with every edge of this path. An upper bound can be established from only the first vertex of each cone. In fact, we only need to upper bound the value $\|f_i t\|_{\hex}$ where $f_i$ is the first vertex of $C_t^i$, hence we look to the previous cone for structure. For $i\in\{1,2,3,4,5\}$, let $y^i\in C_t^i$ be the vertex of $\pi(y,t)_{u^y}$ which maximizes $\| C_{t}^i\cap C_{y^i}^{i-2}\|$. See Figure~\ref{fig:Backtracking} for an example.
\begin{figure}[ht!]
        \centering
        \includegraphics[page=90, scale=1]{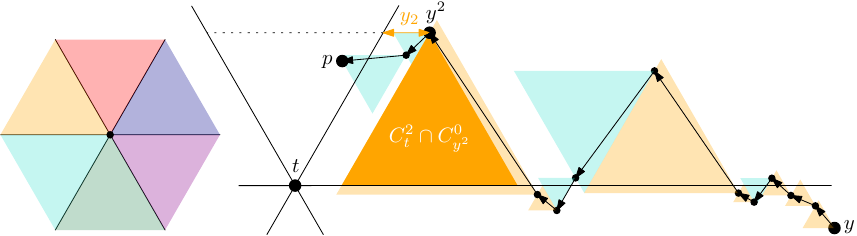}
        \caption{In this example, $p$ is the first vertex of $\pi(y,t)_{v^y}$ in cone $C_t^3$. The highest vertex of $\pi(y,t)_{v^y}$ in $C_t^2$ is $y^2$ by definition. The maximal triangle $C_{t}^2\cap C_{y^2}^{0}$ is coloured orange, and $y_2=\|y^2 t\|_{\hex}-\|C_t^2\cap C_{y^2}^{0}\|$. Observation~\ref{obs:orderedyi} implies that $y^2$ has a greater $y$-coordinate than $p$, thus $\|C_{t}^2\cap C_{y^2}^{0}\|$ provides an upper bound on $\|pt\|_{\hex}$.\label{fig:Backtracking}}
    \end{figure}
\begin{obs}\label{obs:orderedyi}
    For $i\in\{1,2,3,4,5\}$, no vertex of $\pi(y,t)_{y^i}$ is in $\text{int}( C_t^{i+1})$. 
\end{obs}
\begin{subproof}
    The observation holds for $i=1$ since $y^1=y$. Then for $i\in\{2,3,4,5\}$, consider the first vertex $v$ of $ C_t^{i+1}$ on $\pi(y,t)_{u^y}$. Then $y^i$ and $t$ must be on opposite sides of the line $R_v^{i-1}\cup R_v^{i+2}$, which implies $\|y^it\|_{\hex}>\|vt\|_{\hex}$, meaning that $y^i$ appears before $v$ on $\pi(y,t)_{u^y}$.
\end{subproof}
Observation~\ref{obs:orderedyi} implies $y^1,y^2,y^3,y^4,y^5$ appear in order on $\pi(y,t)_{u^y}$. Intuitively, $y^i$ is the head of a {\it long} edge whose tail is in the previous cone $C_t^{i-1}$. Indeed, if the tail was in the same cone as $y^i$, then $\| C_{t}^i\cap C_{y^i}^{i-2}\|$ would not be maximal. We define the distance $y_i$ from $y^i$ to $ C_t^{i+1}$ as $y_i:=\| C_{t}^i\cap C_{y^i}^{i+2}\|$. Equivalently, $y_i=\|y^i t\|_{\hex}-\|C_t^i\cap C_{y^i}^{i-2}\|$. By Assumption~\ref{asm:q}, the $y$-path must pass through cones $C_t^1,...,C_t^5$, hence $y_i$ is well-defined for $i\in\{1,2,3,4,5\}$. Also notice that the region $\text{int}( C_{t}^i\cap C_{y^i}^{i-2})$ is empty since $y^i$ is necessarily the head of an edge from $ C_t^{i-1}$. Next, for each $i\in\mathbb{Z}_6$, we let $Y_i$ be the value of $\|ut\|_{\hex}$ where $u$ is the first vertex of $\pi(y,t)_{u^y}$ with $u\in C_t^i$. If no such vertex $u\in C_t^i$ exists, then $Y_i:=0$. First, we have $Y_1=\|yt\|_{\hex}\leq 1-y_0$. Next, for $i\in\{1,2,3,4,5\}$, cone $ C_t^{i+1}$ can only be visited from cone $ C_t^{i}$ on $\pi(y,t)_{u^y}$, therefore we see that for $i\in \{1,2,3,4,5\}$,
\begin{align}
   y_0+Y_1\leq 1 \qquad\text{and}\qquad y_{i}+Y_{i+1}\leq y_i+ \| C_{t}^{i}\cap C_{y^i}^{i-2}\|=\|y^it\|_{\hex}\leq Y_i\label{eq:Ys}
\end{align}
We include these inequalities in our system which will eventually be proven infeasible. Next we will discuss the second candidate path: the $x$-path. Recall $sx$ is the edge in $C_s^4\cap C_t^0$. We will define the analogous variables for $x$ as we did for $y$. Let $u^x v^x$ be the first edge of $\pi(x,t)$ that crosses $R_t^0$. As before, Lemma~\ref{lem:adj} implies that for some $i\in\{0,1\}$, we have $u^x\in C_t^i$ and $v^x\in C_t^{1-i}$. If no such crossing exists, then set $u^x:=t$ and $v^x:=t$. For $i\in\{0,5,4,3,2\}$, let $x^i\in C_t^i$ be the vertex of $\pi(x,t)_{u^x}$ which maximizes $\| C_{t}^i\cap C_{x^i}^{i+2}\|$. We define the distance $x_i$ from $x^i$ to $ C_t^{i-1}$ as $x_i:=\| C_{t}^i\cap C_{x^i}^{i-2}\|$. If no such $x^i$ exists, then define $x_i:=0$.  However, we will later argue in Claim~\ref{obs:x3} that the vertices $x^0,x^5,x^4,x^3$ must exist, otherwise the $x$-path is sufficiently short to force a contradiction. Observe that for $i\in\{5,4,3,2\}$ the region $ C_{t}^i\cap C_{x^i}^{i+2}$ is empty since $x^i$ is the head of an edge from $ C_t^{i+1}$. Next, for each $i\in\mathbb{Z}_6$, we let $X_i$ be the value of $\|ut\|_{\hex}$ where $u$ is the first vertex of $\pi(x,t)_{u^x}$ with $u\in C_t^i$. If no such vertex $u$ exists, then $X_i:=0$. Since greedy paths decrease in $\|\cdot\|_{\hex}$-norm, then we have $X_0=\|xt\|_{\hex}\leq 1$. Next, for $i\in\{5,4,3,2,1\}$, if $v$ is the first vertex of $\pi(x,t)_{u^x}$ in $C_t^i$, then its predecessor $u$ must be in cone $C_t^{i+1}$, therefore we see that for $i\in \{5,4,3,2,1\}$,
\begin{align}
   X_0\leq 1 \quad\text{and}\quad x_{i+1}+X_{i}\leq x_{i+1} +\| C_{t}^{i+1}\cap C_{x^{i+1}}^{i+3}\|=\|x^{i+1}t\|_{\hex}\leq X_{i+1}.\label{eq:Xs}
\end{align}
Intuitively, the only vertices in cone $C_t^i$ that are important to us are $x^i$ and $y^i$. Now we consider possible paths from $s$ to $t$. We could take the $y$-path from $s$ to $v^y$, then apply induction at $v^y$. By Lemma~\ref{lem:greedyPotential}, this yields
\begin{align}
    5-2.5y_0 &\leq d_G(s,t) \nonumber\\
    &\leq \|sy\|_{\hex}+\|\pi(y,t)_{v^y}\|_{\hex} +d_G(v^y,t) \nonumber\\
    &\leq (1-y_1)+(Y_1+Y_2+Y_3+Y_4+Y_5+Y_0) +5||v^yt||_{\thex}.\label{eq:YPath}
\end{align}
On the other hand, we could instead follow the $x$-path from $s$ to $v^x$, then apply induction at $v^x$. Similarly, we obtain
\begin{align}
    5-2.5y_0 &\leq d_G(s,t) \nonumber\\
    &\leq \|sx\|_{\hex}+\|\pi(x,t)_{v^x}\|_{\hex} +d_G(v^x,t) \nonumber\\
    &\leq (1-y_0-x_0)+(X_0+X_5+X_4+X_3+X_2+X_1) +5||v^xt||_{\thex}.\label{eq:XPath}
\end{align}
While we have proven several inequalities (\ref{eq:firstAssumption}),(\ref{eq:Ys}),(\ref{eq:Xs}),(\ref{eq:YPath}),(\ref{eq:XPath}), they are not yet strong enough to form an infeasible system. First, notice that since $R$ may not contain vertices, then $u^y v^y$ must cross both $B_t^0$ and $B_t^1$. By the following lemma, this implies that the edge $ u^y v^y$ makes significant progress towards $t$. 
\begin{lem}\label{lem:RCrossing}
    Let $u,v,t\in P$ with edge $uv$ contained in $\nabla_u^t$. Assume $uv$ crosses both $B_t^0$ and $B_t^1$, then $2\|vt\|_{\thex}\leq \|u v\|_{\hex}$.
\end{lem}

\begin{proof}
    Suppose without loss of generality that $u\in  C_t^0$ with $u$ to the left of $B_t^0$, and $v\in  C_t^1$ with $v$ above $B_t^1$. See Figure~\ref{fig:RCrossing}. Consider the point $t':= B_t^0\cap R_v^5$. Since $u$ is to the left of $B_t^0$, then the segment $vt'$ is contained in $\nabla_u^v$, hence $\|vt'\|_{2}\leq \|\nabla_u^v\|=\|uv\|_{\hex}$. On the other hand, consider the point $v':=R_t^1\cap B_v^3$. Then considering that $vt'$ and $v't$ have respective slopes $\sqrt{3}$ and $0$, then we have $2\|vt\|_{\thex}=2\|v't\|_{2}=\|vt'\|_{2}\leq \|uv\|_{\hex}$.
    \begin{figure}[ht!]
        \centering
        \includegraphics[page=81, scale=0.9]{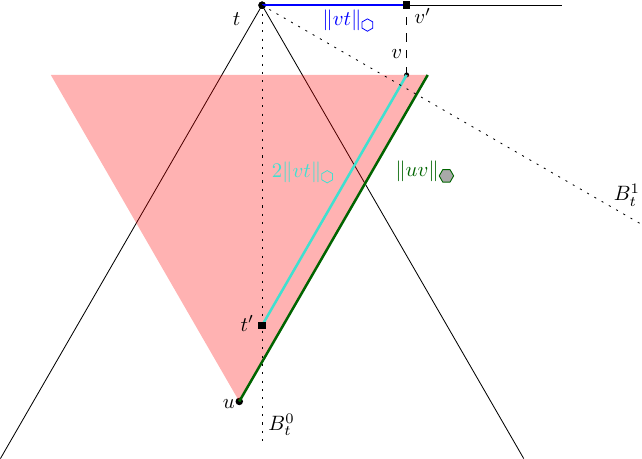}
        \caption{Lemma~\ref{lem:RCrossing} where $u\in  C_t^0$ with $u$ to the left of $B_t^0$, and $v\in  C_t^1$ with $v$ above $B_t^1$. \label{fig:RCrossing}}
    \end{figure}
\end{proof}

Lemma~\ref{lem:RCrossing} quantifies the toll required for an edge to cross the empty region. Intuitively, the distance to $t$ is halved each time the empty region $R$ is crossed. In order to also apply Lemma~\ref{lem:RCrossing} to $u^x v^x$, we require two observations, which follow from Assumptions~\ref{asm:q} and~\ref{asm:x}.
\begin{obs}\label{obs:x1y0}
If $x^1$ exists, then $\|x^1t\|_{\hex}\leq 1-y_0$. 
\end{obs}
\begin{subproof}
    Suppose on the contrary that $\|x^1t\|_{\hex}>1-y_0$. See Figure~\ref{fig:x1y0}. Since greedy paths decrease in $\|\cdot\|_{\hex}$-norm to their destination, then $\|y^it\|_{\hex}$ is monotonically decreasing, whereas $\|x^it\|_{\hex}$ is monotonically increasing for $i\in\{1,2,3,4,5\}$. This means that $\|y^it\|_{\hex}\leq \|y^1t\|_{\hex}\leq 1-y_0 <\|x^1t\|_{\hex}\leq \|x^it\|_{\hex}$. In particular, the empty regions $ C_{t}^i\cap C_{x^i}^{i+2}$ for $i\in\{2,3,4\}$ imply that $Y_3\leq x_2$, $Y_4\leq x_3$, and $Y_5\leq x_4$. Assumption~\ref{asm:q} implies $Y_5\geq y_0$, therefore (\ref{eq:Ys}) implies $\min(x_2,x_3,x_4)\geq y_0$. However this yields $1-y_0< \|x^1t\|_{\hex}< 1-x_0-x_5-x_4-x_3-x_2\leq 1-3y_0$, which implies $y_0<0$. This is a contradiction, hence we may assume $\|x^1t\|_{\hex}\leq 1-y_0$ for the rest of the proof.
\begin{figure}[ht!]
    \centering
    \includegraphics[page=84, scale=1]{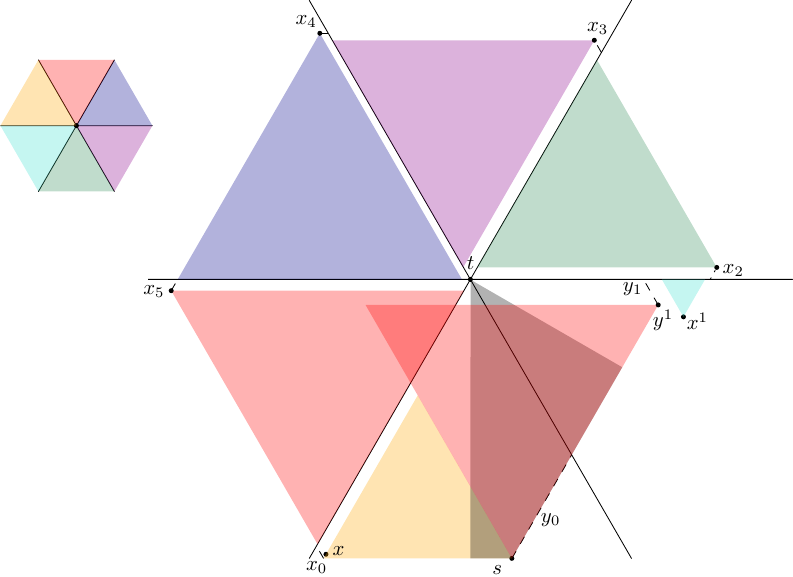}
    \caption{Observation~\ref{obs:x1y0}: From Assumption~\ref{asm:q}, there cannot exist $x^1$ such that $\|x^1t\|>1-y_0$. \label{fig:x1y0}}
\end{figure}
\end{subproof}
\begin{obs}\label{obs:noT}
There is no edge $uv$ on $\pi(x,t)_{v^x}$ with $u\in  C_t^5$ and $v\in T$.
\end{obs}
\begin{subproof}
    If there were such an edge $uv$ on $\pi(x,t)_{v^x}$ with $u\in  C_t^5$ and $v\in T$, the empty region $\text{int}(\nabla_u^v)$ contradicts Assumption~\ref{asm:x}. See Figure~\ref{fig:cone5toT}.
    \begin{figure}[ht!]
        \centering
        \includegraphics[page=83, scale=1]{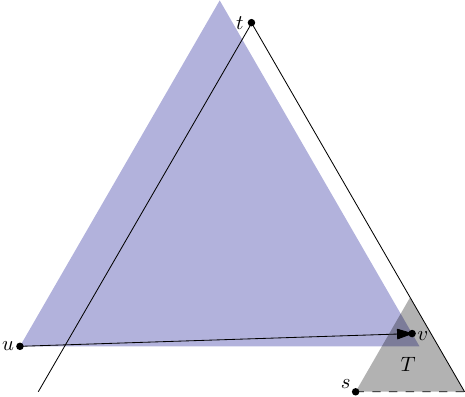}
        \caption{Observation~\ref{obs:noT}: If $uv$ is an edge from $C_t^5$ to $T$, then the region $ C_v^4\cap C_t^0$ contains no vertices in its interior, contradicting Assumption~\ref{asm:x}. \label{fig:cone5toT}}
    \end{figure}
\end{subproof}

Observations~\ref{obs:x1y0} and~\ref{obs:noT} together imply that the edge $u^x v^x$ crosses both $B_t^0$ and $B_t^1$ since the interior of region $R$ is empty by Lemma~\ref{lem:empty}. Therefore Lemma~\ref{lem:RCrossing} applies to $u^x v^x$ as well. Furthermore, we can ensure that $\|v^x t\|_{\hex}\leq x_0$ by the following claim.

\begin{clm}\label{obs:Cone0x0} 
There is no vertex to the left of $B_t^0$ in the interior of $ C_x^2$. 
\end{clm}
\begin{subproof}
    Assume there is such a vertex $p\in C_x^2$ to the left of $B_t^0$. Note that we must have $p\in C_s^3$ since $C_x^2\cap C_s^4\subseteq \nabla_s^x$ and $\nabla_s^x$ contains no vertices. This means that $p$ is above $\nabla_s^y$, yielding $\|pt\|_{\hex}\leq y_1$. See Figure~\ref{fig:emptyx0}. 
    \begin{figure}[ht!]
        \centering
        \includegraphics[page=74, scale=1]{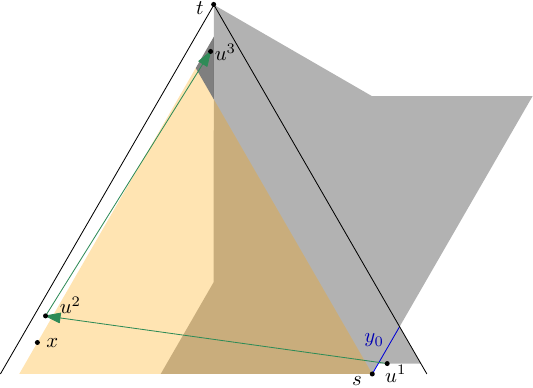}
        \caption{If there is a vertex to the left of $B_t^0$ in $ C_x^2$ (the dark grey region), then we can form a path that passes through $s,u^1,u^2,u^3,t$. Induction is used to bound $d_G(s,u^1)$ and $d_G(u^3,t)$, whereas $u^1u^2$ and $u^2u^3$ are edges in $G$ (in green). \label{fig:emptyx0}}
    \end{figure}
    Let $u^1$ be the vertex of $T:= C_t^0\cap C_s^2$ minimizing $\|u^1t\|_{\hex}$. Then we take the edge $u^1u^2$ in $ C_{u^1}^4$. By Assumption~\ref{asm:x}, we have $u^2\in C_t^0$. If $u^2\in C_x^2$, then set $u^3:=u^2$. Otherwise take the edge $u^2u^3$ in $ C_{u^2}^2$. We must have $u^3\in C_s^3$, hence we apply induction from $u^3$ to $t$. This yields
    \begin{align}
        5-2.5y_0 &<d_G(s,t) \nonumber\\
        &\leq d_G(s,u^1)+\|u^1u^2\|_{\hex}+\|u^2u^3\|_{\hex}+d_G(u^3,t) \nonumber\\
        &\leq 5y_0+1+1+5y_1.\label{eq:emptyx0}
    \end{align}
    However, the system of inequalities (\ref{eq:firstAssumption},\ref{eq:emptyx0}) is infeasible. Therefore we may assume that the region $ C_x^2\cap C_t^0$ to the left of $B_t^0$ is empty for the remainder of the proof.
\end{subproof}
In order to further restrict the set of solutions to our linear system of inequalities towards infeasibility, we will consider several cases with more granularity. We will first consider the direction in which the edge $u^y v^y$ crosses the boundary $R_t^0$ in order to strengthen our upper bound on $\|v^y t\|_{\thex}$. Recall that $5\|v^y t\|_{\thex}$ accounts for the path from $v^y$ to $t$ by induction, hence it is an important quantity to reduce. We will apply the same analysis to $u^x v^x$ to also upper bound $\|v^x t\|_{\thex}$, then lastly we will consider the trade-off between the $x$-path and $y$-path. In Figure~\ref{fig:ExampleT6}, we show an example of the $x$-path and $y$-path that compete for space.
\begin{figure}[ht!]
        \centering
        \includegraphics[page=88, scale=0.8]{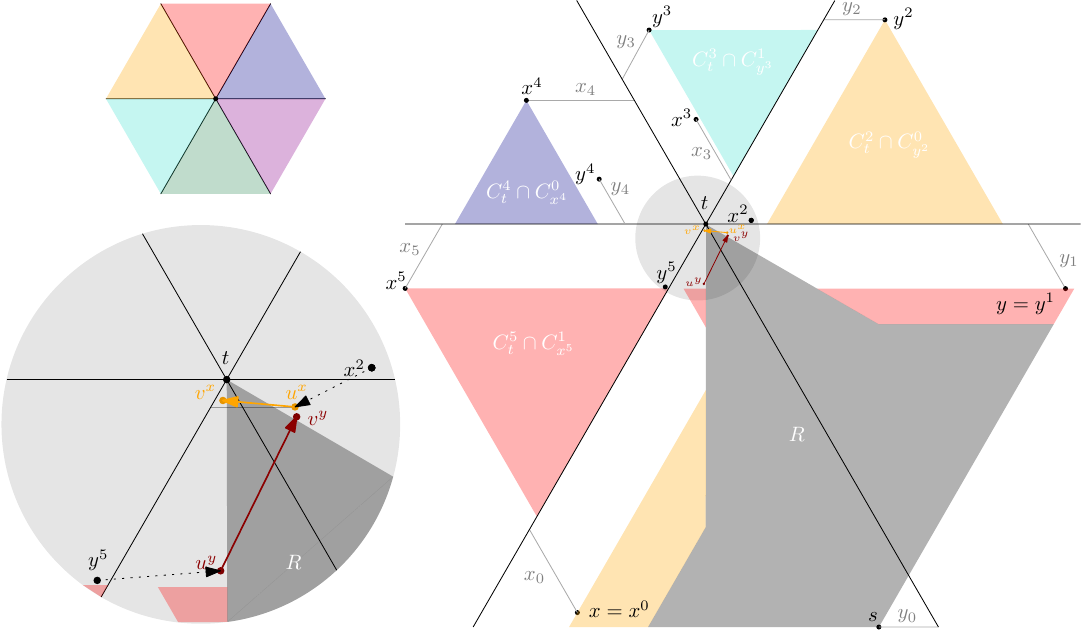}
        \caption{Top Left: The colour coding for edges based on direction. Right: The grey region $R$ contains no vertices in its interior by Lemma~\ref{lem:empty}. The triangular regions are coloured according to their edge direction. This example shows the case when $v^y\in C_t^1$ (case $Y1$), $v^x \in C_t^0$ (case $X0$), and $j=3$ (case J3). Notice that the empty turquoise region $C_t^3\cap C_{y^3}^1$ constrains the position of $x^3$, forcing $x^2$ to be considerably closer to $t$. Lower Left: This enlarged view shows the edges that cross the empty grey region $R$. Notice that the edges $u^xv^x$ and $u^yv^y$ both make considerable progress towards $t$ (quantified in Lemma~\ref{lem:RCrossing}).  \label{fig:ExampleT6}}
    \end{figure}
    
Before proceeding to the cases, we show in Claim~\ref{obs:x3} that the $x$-path must visit the cone $C_t^3$.

\begin{clm}\label{obs:x3}
The path $\pi(x,t)_{v^x}$ must contain a vertex in $ C_t^3$.
\end{clm}
\begin{subproof}

If the claim did not hold, then $u^x v^x$ would necessarily cross $R$ in the counterclockwise direction. See Figure~\ref{fig:EnterCone3}. 
    \begin{figure}[ht!]
        \centering
        \includegraphics[page=70, scale=1]{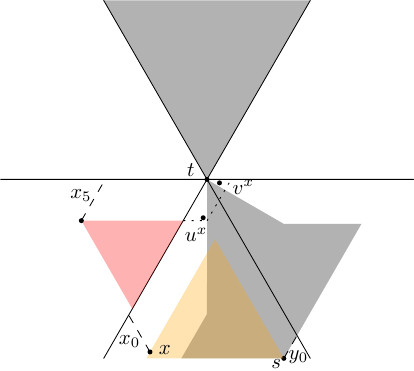}
        \caption{Claim~\ref{obs:x3}: If $\pi(x,t)_{v^x}$ does not enter $C_t^3$, then we can form an infeasible system of linear inequalities. \label{fig:EnterCone3}}
    \end{figure}
    If $\pi(x,t)_{v^x}\cap C_t^5\neq\emptyset$, then we use Lemma~\ref{lem:greedyPotential} and apply induction at $v^x$ to obtain
    \begin{align}
        5-2.5y_0 &< d_G(s,t)\nonumber\\
        &\leq \|sx\|_{\hex}+\|\pi(x,t)_{v^x}\|_{\hex} + d_G(v^x,t) \nonumber\\
        &\leq (1-x_0-y_0) + (X_0 + X_5 + X_4) + 5||v^xt||_{\hex} \nonumber\\
        &\leq (1-x_0-y_0) + 1 + (1-x_0) + (1-x_0-x_5) + 5\min(x_0,\frac{x_5}{2}). \label{eq:v1}
    \end{align}
    However, the linear system (\ref{eq:firstAssumption},\ref{eq:v1}) is infeasible. On the other hand if $\pi(x,t)_{v^x}\cap C_t^5=\emptyset$, then 
    \begin{align*}
        d_G(s,t) &\leq \|sx\|_{\hex}+\|\pi(x,t)_{v^x}\|_{\hex} + d_G(v^x,t) \\
        &\leq (1-x_0-y_0) + X_0  + 5||v^xt||_{\hex} \\
        &\leq (1-x_0-y_0) + 1 + 5x_0\\
        &\leq 4-y_0 &\text{since $x_0\leq 1/2$,}\\
        &\leq 5-2.5y_0 &\text{since $y_0\leq 1/2$.}
    \end{align*}
    which is a contradiction. To summarize, $\pi(x,t)_{v^x}$ must enter $C_t^3$.
\end{subproof}
First, notice that $u^yv^y$ can cross $R_t^0$ either in the clockwise or counterclockwise direction:

\underline{Case Y0: $v^y\in C_t^0$.} Then we must have $u^y\in C_t^1$. We obtain the following two inequalities:
\begin{align}
    ||v^yt||_{\thex} &\leq y_1, && \text{by empty region $\text{int}(\nabla_s^y)$.}\\
    ||v^yt||_{\thex} &\leq \frac{1}{2}\min(y_2,y_3,y_4), && \text{by Assumption~\ref{asm:q} and Lemma~\ref{lem:RCrossing}.}
\end{align}

\underline{Case Y1: $v^y\in C_t^1$.} Then we must have $u^y\in C_t^0$. We obtain
\begin{align}
    ||v^yt||_{\thex} &\leq x_0, && \text{by Claim~\ref{obs:Cone0x0}.}\\
    ||v^yt||_{\thex} &\leq \frac{Y_0}{2}, && \text{by Lemma~\ref{lem:RCrossing}.}
\end{align}
Similarly, $u^xv^x$ can also cross $R_t^0$ in one of two possible ways. 

\underline{Case X0: $v^x\in C_t^0$.} Then we have
\begin{align}
    ||v^xt||_{\thex} &\leq y_1, && \text{by Observation~\ref{obs:x1y0} and empty region $\text{int}(\nabla_s^y)$.}\\
    ||v^xt||_{\thex} &\leq \frac{X_1}{2}, && \text{by Lemma~\ref{lem:RCrossing}.}
\end{align}

\underline{Case X1: $v^x\in C_t^1$.} Then $u^y\in C_t^0$. We have
\begin{align}
    ||v^xt||_{\thex} &\leq x_0, && \text{by Claim~\ref{obs:Cone0x0} and~\ref{obs:noT}.}\\
    ||v^xt||_{\thex} &\leq \frac{1}{2}\min(x_5,x_4,x_3), && \text{by Claim~\ref{obs:x3} and Lemma~\ref{lem:RCrossing}.}
\end{align}

Next, we will derive several useful inequalities that reflect the balance between the two candidate paths. Intuitively, if $\pi(y,t)_{v^y}$ is long, then the empty regions $C_t^i\cap C_{y^i}^{i-2}$ must be large, leaving no room for the other path $\pi(x,t)_{v^x}$. We formalize this trade-off as follows. Let $j$ be the greatest index in $\{2,3,4,5\}$ such that $\| C_t^j\cap C_{y^j}^{j-2}\|>x_j$. For example when $j=3$, see Figure~\ref{fig:ExampleT6}.

\underline{Case: No such $j$ exists.} Then by definition we must have 
\begin{align}
\label{noj}
    Y_{i+1} \leq \| C_t^i\cap C_{y^i}^{i-2}\| \leq x_i \quad \text{for } 2\leq i\leq 5.
\end{align}

\underline{Case: $j\in\{2,3,4,5\}$.} As $j$ is maximal, then $\| C_t^i\cap C_{y^i}^{i-2}\| \leq x_i$ for $j<i\leq 5$. Next, since $\| C_t^j\cap  C_{y^j}^{j-2}\| > x_j$, then we must have $\|y^kt\|_{\hex}>\|x^kt\|_{\hex}$ for $2\leq k\leq j$. By the empty regions $ C_{t}^i\cap C_{y^i}^{i-2}$, we must have 
\begin{align}
\label{aj}
    Y_{i+1} \leq \| C_t^i\cap C_{y^i}^{i-2}\| \leq x_i \quad \text{for }j<i\leq 5,\quad\text{ and }\quad X_{k-1}\leq y_k  \quad \text{for }2\leq k\leq j.
\end{align}
We denote the cases where $j\in\{2,3,4,5\}$ as J2, J3, J4, J5, and observe that for each case, we yield four inequalities from \eqref{aj}.

Under Assumptions~\ref{asm:q} and~\ref{asm:x} we have the following possibility space: $\{\text{Y0, Y1}\}\times\{\text{X0, X1}\}\times \{\eqref{noj},\text{ J2, J3, J4, J5}\}$. In total there are $2\times 2\times 5=20$ possibilities, and in every case, the inequalities \eqref{eq:firstAssumption}, \eqref{eq:Ys}, \eqref{eq:Xs}, \eqref{eq:YPath}, \eqref{eq:XPath} also hold. It can be shown that the linear system of inequalities is infeasible for $16$ of the $20$ possibilities. In Section~\ref{sec:Poss}, we show that for the remaining four cases, we can define a third candidate path which is sufficiently short for a contradiction. These remaining four cases are (X1, Y1, J2), (X1, Y1, J3), (X1, Y1, J4), (X0, Y0, J3).

\section{Remaining Cases for Upper Bound Proof}
 In Section~\ref{sec:notq}, we prove Theorem~\ref{thm:sr} when Assumption~\ref{asm:q} does not hold. Then, Section~\ref{sec:onlyq} proves Theorem~\ref{thm:sr} when Assumption~\ref{asm:q} holds but Assumption~\ref{asm:x} does not. Finally, in Section~\ref{sec:Poss}, we take care of the remaining four cases involved in proving Theorem~\ref{thm:sr} under Assumptions~\ref{asm:q} and~\ref{asm:x}.

\subsection{Proving Theorem~\ref{thm:sr} without Assumption~\ref{asm:q}}\label{sec:notq}
    In this section, we assume that Assumption~\ref{asm:q} does {\bf not} hold. This means that if $q\in \pi(y,t)_{v^y}\cap C_t^5$, then $||\nabla_t^q||\leq y_0$. See Figure~\ref{fig:5ConeMostlyEmpty}.
    \begin{figure}[ht!]
        \centering
        \includegraphics[page=51, scale=1]{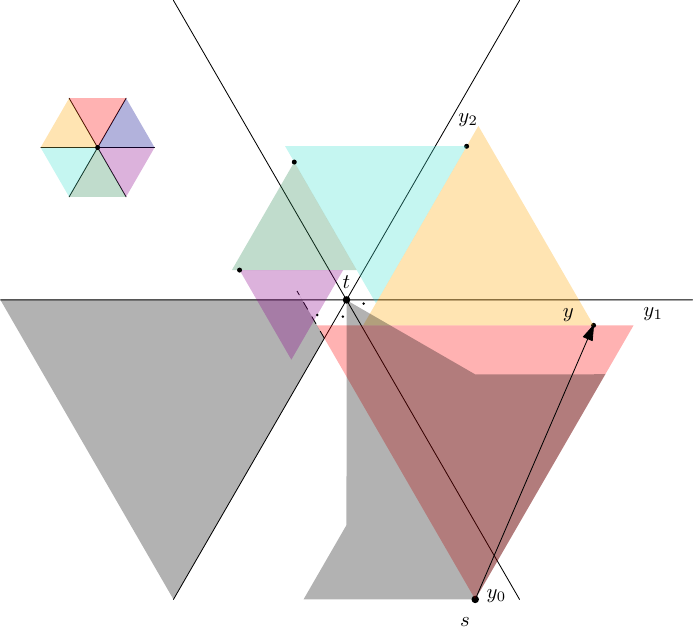}
        \caption{When Assumption~\ref{asm:q} does not hold, then the shaded region in $ C_t^5$ does not contain vertices of $\pi(y,t)_{v^y}$.\label{fig:5ConeMostlyEmpty}}
    \end{figure} We consider the following three cases:
    
    \underline{Case $\pi(y,t)_{v^y}\cap C_t^2=\emptyset$:}  Then $u^yv^y$ crosses from $ C_t^1$ to $ C_t^0$ and $Y_i=0$ for $i\neq 1$. By Lemma~\ref{lem:greedyPotential}, we obtain
    \begin{align}
        5-2.5y_0 &<d_G(s,t) \nonumber\\
        &\leq \|sy\|_{\hex}+\|\pi(y,t)_{v^y}\|_{\hex}+d_G(v^y,t) \nonumber\\
        &\leq (1-y_1)+Y_1+5\|v^yt\|_{\hex} \nonumber\\
        &\leq (1-y_1) + (1-y_0)+5y_1.\label{eq:case1}
    \end{align} 
    It can be verified that the linear system (\ref{eq:firstAssumption},\ref{eq:case1}) is infeasible.
    
    \underline{Case $v^y\in C_t^0$:}
    We can assume $\pi(y,t)_{v^y}$ visits $ C_t^2$, so $y^2$ is well-defined. Recall $y^2\in C_t^2\cap\pi(y,t)_{v^y}$ and $y^2$ maximizes $\| C_{t}^2\cap C_{y^2}^{0}\|$. Also note that we defined $y_2:=\| C_{t}^2\cap C_{y^2}^{4}\|$. Then $\|u^yt\|_{\hex}\leq y_2$, so by Lemma~\ref{lem:RCrossing}, $\|v^yt\|_{\thex}\leq \frac{y_2}{2}$. In addition, $u^y\in C_s^3$, hence we also have $\|v^yt\|_{\hex}\leq y_1$. We can explicitly upper bound $Y_i$ as follows:
    \begin{align}
        5-2.5y_0 &<d_G(s,t) \nonumber\\
        &\leq \|sy\|_{\hex}+\|\pi(y,t)_{v^y}\|_{\hex}+d_G(v^y,t) \nonumber\\
        &\leq (1-y_1)+(Y_1+Y_2+Y_3+Y_4+ Y_5+Y_0)+5\|v^yt\|_{\thex} \nonumber\\
        &\leq (1-y_1) + (1-y_0) + (1-y_0-y_1) + (1-y_0-y_1-y_2) \nonumber\\
        &\phantom{\leq}\quad+ (1-y_0-y_1-y_2) + y_0 + y_1 +5\min(y_1,\frac{y_2}{2}).\label{eq:case1v}
    \end{align}
    It can be verified that the linear system (\ref{eq:firstAssumption},\ref{eq:case1v}) is infeasible.

    \underline{Case $v\in C_t^1$:} 
    Then $u^y\in C_t^0$. Since Assumption~\ref{asm:q} does not hold, then $u^y\in C_s^3$, forcing $\|u^yt\|_{\hex}\leq y_1$. By Lemma~\ref{lem:RCrossing}, we have $\|v^yt\|_{\thex}\leq \frac{y_1}{2}$.
    \begin{align}
        5-2.5y_0 &<d_G(s,t) \nonumber\\
        &\leq \|sy\|_{\hex}+\|\pi(y,t)_{v^y}\|_{\hex}+d_G(v^y,t) \nonumber\\
        &\leq (1-y_1)+(Y_1+Y_2+Y_3+Y_4+Y_5+Y_0)+5\|v^yt\|_{\hex} \nonumber\\
        &\leq (1-y_1) + (1-y_0) + (1-y_0-y_1) + 2(1-y_0-y_1-y_2) + y_0 + y_1 + 5\frac{y_1}{2}.\label{eq:case2v}
    \end{align}
    It can be verified that the linear system (\ref{eq:firstAssumption},\ref{eq:case2v}) is infeasible, concluding this section.

\subsection{Proving Theorem~\ref{thm:sr} using only Assumption~\ref{asm:q}}\label{sec:onlyq}
     This section considers the case when Assumption~\ref{asm:q} holds but Assumption~\ref{asm:x} does not. Then there exists a vertex $u$ in $T:= C_t^0\cap C_s^2$ with no vertices strictly in $ C_u^4\cap C_t^0$. Let $u^0$ be the vertex in $T$ minimizing $\|u^0t\|_{\hex}$. Note that $u^0:=s$ if the interior of $T$ is empty. Let $h:=1-\|u^0t\|_{\hex}$ and $d:=\|su^0\|-h$. See Figure~\ref{fig:dAndh}. Since $u^0\in T$, we have 
    \begin{align}
        d+h\leq y_0.\label{eq:dAndh}
    \end{align}
    
    \begin{figure}[ht!]
        \centering
        \includegraphics[page=77, scale=1]{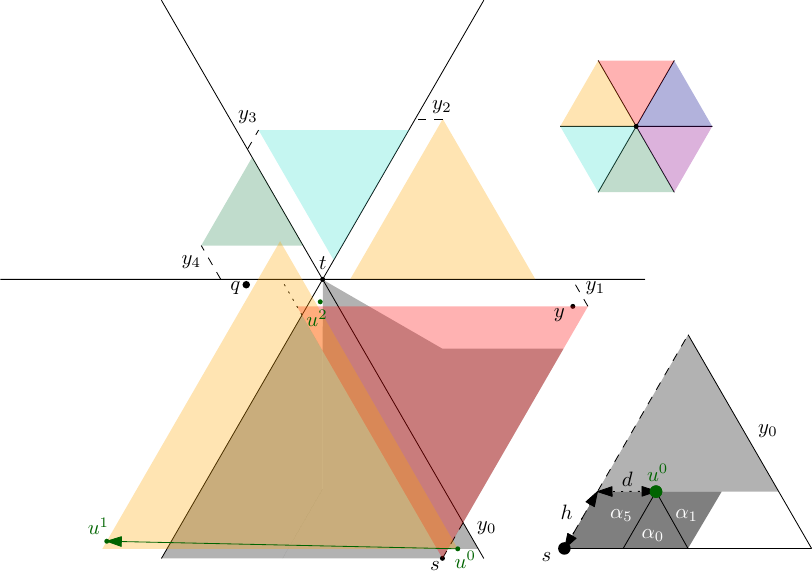}
        \caption{When $ C_s^4\cap C_t^0$ is empty, then we form a path through $s,u^0,u^1,u^2,t$. Definition of $d$ and $h$.\label{fig:dAndh}}
    \end{figure}
    
    We proceed depending on whether the greedy path $\pi(s,u^0)$ enters the cone $C_{u^0}^2$. 
    
    \underline{Case $ C_{u^0}^2\cap\pi(s,u^0)=\emptyset$:} Let $\alpha_i$ be the maximum value of $\|u^0u\|_{\hex}$ over all vertices $u\in\pi(s,u^0)\cap C_{u^0}^i$. Considering the empty regions around $u^0$ we have $\alpha_5\leq h+d, \alpha_0\leq h, \alpha_1\leq h$, and $\alpha_i=0$ for $i\in\{2,3,4\}$. By Assumption~\ref{asm:q}, there must exist $q\in  C_{s}^4\cap \pi(y,t)_{v^y}$. Notice that we must also have $q\in  C_{u^0}^4$ since $\|u^0t\|_{\hex}\geq 1-y_0\geq \|yt\|_{\hex}\geq \|qt\|_{\hex}$. Therefore we may consider the edge $u^0u^1$ in $ C_{u^0}^4$. See Figure~\ref{fig:dAndh}. By defining $m:=||u^0u^1||_{\hex}$ and $p:=R_t^0\cap R_s^1$, we obtain 
    \begin{align}
        (y_0-d) + m &\leq \|pu^0\|_{\hex}+\|u^0 q\|_{\hex}\nonumber &\text{since $||u^0u^1||_{\hex}\leq ||u^0q||_{\hex}$,}\\
        &=\|pq\|_{\hex}\nonumber &\text{by $\|\cdot\|_{\hex}$ collinearity,}\\
        &\leq \|pt\|_{\hex}+\|tq\|_{\hex}\nonumber &\text{by triangle inequality,}\\
        &\leq 1 + (1-y_0-y_1-y_2-y_3-y_4) &\text{since $q\in\pi(y,t)_{v^y}\cap  C_t^5$.}\label{eq:segments}
    \end{align}
    
    Let $u^2$ be the first point of $\pi(u^1,t)$ in $ C_t^0$. Then necessarily $u^2\in C_s^3$, so $\|u^2t\|_{\hex}\leq y_1$. Let $\beta_i$ be the maximum value of $\|ut\|_{\hex}$ over all vertices $u\in\pi(u^1,t)_{u^2}\cap C_t^i$. By Lemma~\ref{lem:greedyPotential}, we have
    \begin{align}
        5-2.5y_0 &<d_G(s,t) \nonumber\\
        &\leq \|\pi(s,u^0)\|_{\hex}+\|u^0u^1\|_{\hex}+\|\pi(u^1,t)_{u^2}\|_{\hex}+d_G(u^2,t) \nonumber\\
        &\leq (\alpha_5+\alpha_0+\alpha_1) + m + (\beta_5+\beta_4+\beta_3+\beta_2+\beta_1) + 5\|u^3t\|_{\hex} \nonumber\\
        &\leq (h+d)+h+h+m+(y_0-h-d+m)+(y_4+y_0-h-d)+y_4+y_3+y_2+ 5y_1.\label{eq:LongSweepPath}
    \end{align}
    However the system (\ref{eq:firstAssumption}, \ref{eq:dAndh}, \ref{eq:segments}, \ref{eq:LongSweepPath}) is infeasible. Notice that \ref{eq:LongSweepPath} holds even if $u^1\in C_t^4$. We re-bracket for clarity:
    \begin{align*}
        5-2.5y_0 &<d_G(s,t)\\
        &\leq \|\pi(s,u^0)\|_{\hex}+\|u^0u^1\|_{\hex}+\|\pi(u^1,t)_{u^2}\|_{\hex}+d_G(u^2,t)\\
        &\leq (\alpha_5+\alpha_0+\alpha_1) + m + (\beta_5+\beta_4+\beta_3+\beta_2+\beta_1) + 5\|u^3t\|_{\hex}\\
        &\leq (h+d)+h+h+m+(y_0-h-d)+(m+y_4+y_0-h-d)+y_4+y_3+y_2+ 5y_1.
    \end{align*}
    This concludes the case when $C_{u^0}^2\cap\pi(s,u^0)=\emptyset$.

    \underline{Case $C_{u^0}^2\cap\pi(s,u^0)\neq\emptyset$:} Then let $v^0$ be the first vertex of $\pi(s,u^0)$ in $C_{u^0}^2$. Next, let $v^1$ be the first vertex of $\pi(v^0,t)$ in $ C_t^0$, and define $T_i$ to be the maximum value of $\|ut\|_{\hex}$ over all vertices $u\in\pi(v^0,t)_{v^1}\cap C_t^i$. Notice that $\|u^0v^0\|_{\hex}\leq h$ and $u^0\in C_t^0$. Therefore $T_i\leq h$ for $i\in\{2,3,4,5\}$. Furthermore, we have $T_1=\|v^0t\|_{\hex}\leq\|v^0 u^0\|_{\hex}+\|u^0t\|_{\hex}\leq h+(1-h)=1$. We have
    \begin{align}
        5-2.5y_0 &<d_G(s,t) \nonumber\\
        &\leq \|\pi(s,u^0)_{v^0}\|_{\hex} + \|\pi(v^0,t)_{v^1}\|_{\hex}+d_G(v^1,t) \nonumber\\
        &\leq (\alpha_5+\alpha_0+\alpha_1)+ (T_1+T_2+T_3+T_4+T_5)+5\|v^1t\|_{\hex} \nonumber\\
        &\leq (h+d)+h+h + 1 + h+h+h+h + 5y_1.\label{eq:DegLongSweep}
    \end{align}
    However the system (\ref{eq:firstAssumption}, \ref{eq:dAndh}, \ref{eq:segments}, \ref{eq:DegLongSweep}) is infeasible. This concludes the proof.

\subsection{Remaining Possibilities under Assumptions~\ref{asm:q} and~\ref{asm:x}}\label{sec:Poss}

Recall from Section~\ref{sec:UpperBound} that we consider the $20$ possibilities from the space $\{\text{Y0, Y1}\}\times\{\text{X0, X1}\}\times \{\eqref{noj},\text{ J2, J3, J4, J5}\}$. In each of these possibilities, the inequalities (\ref{eq:firstAssumption}, \ref{eq:Ys}, \ref{eq:Xs}, \ref{eq:YPath}, \ref{eq:XPath}) hold. However, only $16$ of the $20$ possibilities have infeasible systems. We therefore consider the remaining four cases more carefully: (X1, Y1, J2), (X1, Y1, J3), (X1, Y1, J4), (X0, Y0, J3).

\underline{Case (X1, Y1, J2):} First, if $||v^yt||_{\hex}\leq \frac{y_1}{2}$, then the system is infeasible since $||v^yt||_{\thex}\leq ||v^yt||_{\hex}$. Otherwise, Lemma~\ref{lem:RCrossing} implies $||u^yt||_{\hex}>y_1$, meaning that $u^y\in C_s^4\cap\in C_y^5$. See Figure~\ref{fig:X1Y1case}.
\begin{figure}[ht!]
    \centering
    \includegraphics[page=75, scale=1]{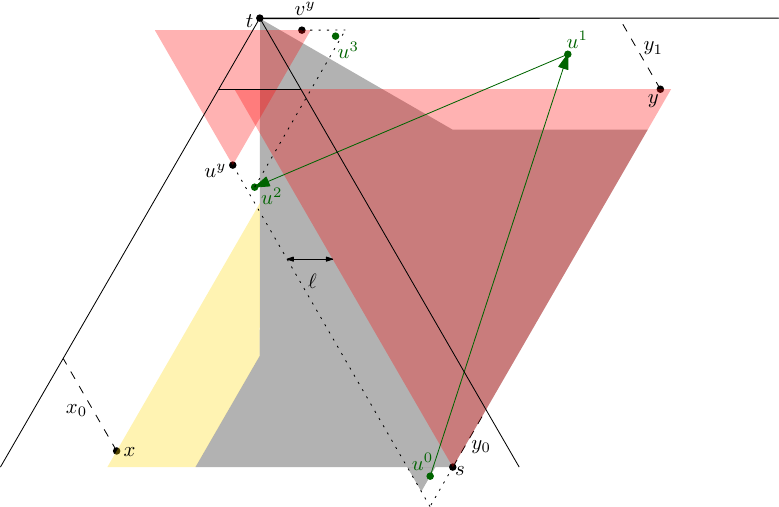}
    \caption{Here we illustrate the case X1,Y1 when $||u^yt||_{\hex}>y_1$. This implies $u^y\in C_y^5$, which allows a shortcut to be formed through vertices $s,u^0,u^1,u^2,u^3,t$. \label{fig:X1Y1case}}
\end{figure}
Let $\ell:=\| C_s^4\cap C_{u^y}^2\|$. Since $u^y v^y$ is an edge and $u^y\in C_t^0$, we must have 
\begin{align}
    Y_0\geq\|u^yv^y\|_{\hex}\geq y_0+\ell+\frac{||v^yt||_{\hex}}{2}\geq y_0+\ell+\frac{||v^yt||_{\thex}}{2}.\label{eq:Y0}
\end{align} 
Now we describe how $u^y$ provides a shortcut to $t$. Let $u^0$ be the vertex in the triangle $ C_{u^y}^1\cap C_s^5$ that minimizes $||u^0u^y||_{\hex}$. If this triangle is empty, then $u^0:=s$. Notice that $||su^0||_{\hex}\leq \ell$. Therefore we use induction to find a path from $s$ to $u^0$, then follow the edge in $ C_{u^0}^3$ to a vertex $u^1$. Notice that $u^1$ cannot be in $ C_t^5$ from the empty region of edge $u^yv^y$. If $\|u^1t\|_{\hex}\leq x_0$, then apply induction (below set $u^2:=u^1,u^3:=u^1$). Otherwise if $u^1\in C_t^0$, then let $u^2:=u^1$. If instead $u^1\in C_t^1$, then follow $ C_{u^1}^5$ to a vertex $u^2\in C_t^0$. Notice that $u^y$ guarantees the existence of $u^2$ and the empty region from edge $u^yv^y$ guarantees $u^2\in C_t^0$. Finally, from $u^2$, follow $\pi(u^2,t)$ until reaching a vertex $u^3\in C_t^1$. Again, the empty triangle $\nabla_{u^y}^{v^y}$ guarantees that $u^3\in C_t^1$. We have 
\begin{align}
    5-2.5y_0 &< d_G(s,t) \nonumber\\
    &\leq d_G(s,u^0)+\|u^0u^1\|_{\hex}+\|u^1u^2\|_{\hex}+\|\pi(u^2,t)_{u^3}\|_{\hex}+d_G(u^3,t) \nonumber\\
    &\leq 5\ell + (1+\ell) + (1+\ell) +(y_0+\ell+x_0) + 5x_0.\label{eq:short1}
\end{align}
It can be verified that the system of case (X1, Y1, J2) is infeasible with the additional inequalities (\ref{eq:Y0}) and (\ref{eq:short1}) in the presence of (\ref{eq:firstAssumption}, \ref{eq:Ys}, \ref{eq:Xs}, \ref{eq:YPath}, \ref{eq:XPath}).

\underline{Case (X1, Y1, J3):} We use the exact same argument as for (X1, Y1, J2). A careful scan of the proof reveals that the argument only uses Y1 for structure.

\underline{Case (X1, Y1, J4):} Again, we use the exact same argument as for (X1, Y1, J2).

\underline{Case (X0, Y0, J3):} We will describe a symmetric argument to the previous three. First, if $||v^xt||_{\hex}\leq \frac{x_0}{2}$, then the system is infeasible since $||v^xt||_{\thex}\leq ||v^xt||_{\hex}$. Otherwise, we must have $||u^xt||_{\hex}>x_0$, meaning that $u^x\in C_x^2$. See Figure~\ref{fig:X0Y0case}.
\begin{figure}[ht!]
    \centering
    \includegraphics[page=76, scale=1]{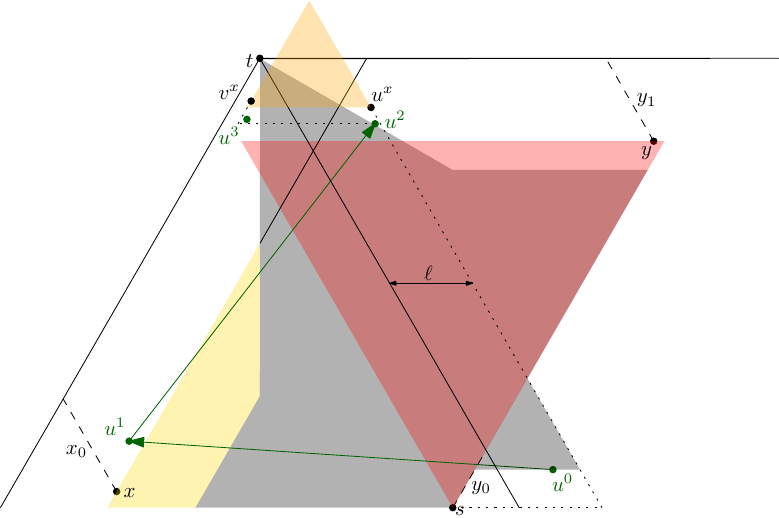}
    \caption{Here we illustrate the case X0,Y0 when $||u^xt||_{\hex}>x_0$. This implies $u^x\in C_x^2$, which allows a shortcut to be formed through vertices $s,u^0,u^1,u^2,u^3,t$. \label{fig:X0Y0case}}
\end{figure}
Let $\ell:=\| C_t^1\cap  C_{u^x}^5\|$. Since $u^x v^x$ is an edge and $u^x\in C_t^1$, we must have 
\begin{align}
    X_1\geq\|u^xv^x\|_{\hex}\geq \ell+\frac{||v^xt||_{\hex}}{2}\geq \ell+\frac{||v^xt||_{\thex}}{2}.\label{eq:X1}
\end{align} 
Now we describe how to use $u^x$ as a shortcut to get to $t$. Let $u^0$ be the vertex in the triangle $ C_{u^x}^0\cap C_s^2$ that minimizes $||u^0u^x||_{\hex}$. If this triangle is empty, then $u^0:=s$. Notice that $||su^0||_{\hex}\leq \ell+y_0$. Therefore we use induction to find a path from $s$ to $u^0$, then follow the edge in $ C_{u^0}^4$ to a vertex $u^1$. Notice that we cannot have $u^2\in C_t^5$ by Assumption~\ref{asm:x} and we cannot have $u^2\in C_t^2$ by the empty region of $u^xv^x$. If $\|u^1t\|_{\hex}\leq y_1$, then apply induction (below set $u^2:=u^1,u^3:=u^1$). If instead $u^1\in C_t^1$, then let $u^2:=u^1$. If instead $u^1\in C_t^0$, then follow $ C_{u^1}^2$ to a vertex $u^2\in C_t^1$. Finally, from $u^2$, follow $\pi(u^2,t)$ until reaching a vertex $u^3\in C_t^0$. By the empty triangle $\nabla_{u^x}^{v^x}$, vertex $u^3$ is well-defined. We have 
\begin{align}
    5-2.5y_0 &< d_G(s,t) \nonumber\\
    &\leq d_G(s,u^0)+\|u^0u^1\|_{\hex}+\|u^1u^2\|_{\hex}+\|\pi(u^2,t)_{u^3}\|_{\hex}+d_G(u^3,t) \nonumber\\
    &\leq 5(\ell+y_0) + (1+\ell) + (1+\ell) +(\ell+y_1) + 5y_1.\label{eq:short2}
\end{align}
It can be verified that the system of case (X0, Y0, J3) is infeasible with the additional inequalities (\ref{eq:X1}) and (\ref{eq:short2}) in the presence of (\ref{eq:firstAssumption}, \ref{eq:Ys}, \ref{eq:Xs}, \ref{eq:YPath}, \ref{eq:XPath}).

\bibliography{DT6}

\end{document}